\documentclass[reqno]{amsart}
\usepackage[letterpaper, total={6.5in, 9in}]{geometry}
\usepackage{amsrefs,color,amssymb}
\usepackage[colorlinks=false]{hyperref}
\usepackage{slashed,dsfont}
\title{Covariant Guiding Laws for Fields} 

\author[M. Derakhshani]{Maaneli Derakhshani} 
\author[M.K.-H.~Kiessling]{Michael K.-H.\ Kiessling} 
\author[A.S.~Tahvildar-Zadeh]{A.\ Shadi Tahvildar-Zadeh} 

\address{Department of Mathematics, Rutgers, The State University of New Jersey, 
110 Frelinghuysen Rd., Piscataway, NJ 08854-8019, United States of America}
\email{\href{mailto:shadi@math.rutgers.edu}{shadi@math.rutgers.edu } {md1485@math.rutgers.edu} {miki@math.rutgers.edu}}

\date{orig. October 19, 2021; this rev. Feb. 12, 2023}

\numberwithin{equation}{section}

\theoremstyle{plain}
\newtheorem{theorem}{Theorem}[section]

\newtheorem{proposition}[theorem]{Proposition}

\theoremstyle{definition}

\newtheorem{example}[theorem]{Example}

\theoremstyle{remark}
\newtheorem{remark}[theorem]{Remark}

\usepackage[draft]{changes}

\def\RR{\mathbb{R}}
\def\SS{\mathbb{S}}

\let\oldmarginpar\marginpar
\renewcommand\marginpar[1]{\-\oldmarginpar[\raggedleft\footnotesize #1]%
{\raggedright\footnotesize #1}}

\newcommand{\cV}{{\mathcal V}}

\newcommand{\refeq}[1]{(\ref{#1})}

\newcommand{\Aset}{{\mathbb A}}
\newcommand{\Bset}{{\mathbb B}}
\newcommand{\Cset}{{\mathbb C}}

\newcommand{\Rset}{{\mathbb R}}

\newcommand{\bJ}{\mathbf{J}}

\newcommand{\bK}{\mathbf{K}}

\newcommand{\bn}{\mathbf{n}}

\newcommand{\bq}{\mathbf{q}}

\newcommand{\bS}{\mathbf{S}}

\newcommand{\bv}{\mathbf{v}}

\newcommand{\bX}{\mathbf{X}}
\newcommand{\bY}{\mathbf{Y}}

\newcommand{\bseq}{\begin{subequation}}
\newcommand{\eseq}{\end{subequation}}
\newcommand{\dal}{\,\raisebox{3pt}{\fbox{}}\,}

\newcommand{\bu}{\mathbf{u}}

\newcommand{\beq}{\begin{equation}}
\newcommand{\eeq}{\end{equation}}

\newcommand{\tr}{\mathrm{tr}\,}

\newcommand{\p}{\partial}
\newcommand{\eps}{\epsilon}

\newcommand{\cA}{\mathcal{A}}
\newcommand{\cB}{\mathcal{B}}
\newcommand{\cF}{\mathcal{F}}

\newcommand{\cC}{\mathcal{C}}
\newcommand{\cD}{\mathcal{D}}

\newcommand{\cH}{{\mathcal H}}

\newcommand{\cK}{{\mathcal K}}

\newcommand{\cL}{{\mathcal L}}
\newcommand{\cM}{{\mathcal M}}
\newcommand{\cN}{{\mathcal N}}
\newcommand{\cP}{{\mathcal P}}

\newcommand{\cS}{{\mathcal S}}
\newcommand{\cT}{\mathcal{T}}

\newcommand{\el}{\ell}
\newcommand{\la}{\lambda}
\newcommand{\La}{\Lambda}
\newcommand{\de}{\delta}
\newcommand{\De}{\Delta}
\newcommand{\al}{\alpha}

\newcommand{\ga}{\gamma}

\newcommand{\ep}{\epsilon}

\newcommand{\om}{\omega}
\newcommand{\Om}{\Omega}
\newcommand{\si}{\sigma}
\newcommand{\Si}{\Sigma}
\newcommand{\nab}{\nabla}
\newcommand{\half}{\frac{1}{2}}

\newcommand{\diag}{\mbox{diag}}

\newcommand{\bna}{\begin{eqnarray}}
\newcommand{\ena}{\end{eqnarray}}
\newcommand{\bea}{\begin{eqnarray*}}
\newcommand{\eea}{\end{eqnarray*}}
\newcommand{\ben}{\begin{enumerate}}
\newcommand{\een}{\end{enumerate}}
\newcommand{\bi}{\begin{itemize}}
\newcommand{\ei}{\end{itemize}}

\newcommand{\stz}[1]{{\textcolor{black}{#1}}}
\newcommand{\newstz}[1]{{\textcolor{black}{#1}}}
\newcommand{\stznew}[1]{{\textcolor{black}{#1}}}
\newcommand{\shadibs}{\textcolor{black}{\mathbf{s}}}

\newcommand{\der}[1]{\textcolor{black}{#1}}

\newcommand{\miki}[1]{\textcolor{black}{#1}}

\begin{document}

\begin{abstract}
After reviewing what is known about the passage from the classical Hamilton--Jacobi formulation of non-relativistic point-particle dynamics to the non-relativistic
quantum dynamics of point particles whose motion is guided by a wave function that satisfies Schr\"odinger's or Pauli's equation, we study the analogous question for
the Lorentz-covariant dynamics of fields on spacelike slices of spacetime. We establish a relationship, between the DeDonder--Weyl--Christodoulou formulation of covariant Hamilton--Jacobi equations for the classical field evolution, and the Lorentz-covariant Dirac-type wave equation proposed by Kanatchikov amended by \stz{our proposed}  
guiding equation for such fields. We show that Kanatchikov's equation is well-posed and generally solvable, and we establish the correspondence between plane-wave solutions of Kanatchikov's equation 
and solutions of the covariant Hamilton--Jacobi equations of  DeDonder--Weyl--Christodoulou. 
\stz{ We propose a 
covariant guiding law for the temporal evolution of fields defined on constant time slices of spacetime, and show that it yields, at each spacetime point, the existence of a finite measure on the space of field values at that point 
that is equivariant with respect to the flow induced \stznew{by} the solution of Kanatchikov's equation 
that is guiding the actual field,} \stznew{so long as it is a plane-wave solution.}
We show that \stz{our guiding law} is local in the sense of Einstein's special relativity, \stz{and therefore it cannot be used to analyze Bell-type experiments.}
\stz{W}e conclude by suggesting directions to be explored in future research.

\begin{center}
{\em Dedicated to Demetrios Christodoulou on the occasion of his 70th birthday}
\end{center}
\end{abstract}
\maketitle

\vspace{-.5truecm}
\section{Introduction and \stz{Statement of Main Results}}

  Action principles have played a crucial role in the formulation of fundamental physical theories. 
  Originally invented for classical mechanics and subsequently generalized to classical field theory, as recalled and advanced in Christodoulou's monograph \cite{C1}, 
classical action principles have also been a point of departure for arriving at quantum mechanics and quantum field theory through some kind of ``quantization'' procedure.
  In this paper we explore a quantization procedure that could establish a link 
between the Hamilton--Jacobi formalism of DeDonder--Weyl--Christodoulou in their formulation of Lorentz-covariant classical 
field theory, and Kanatchikov's \cite{Kan99} covariant quantum wave equation on 
\miki{a space of generic fields}, 
amended by our \stz{proposed} guiding equation 
for the evolution of actual fields on Minkowski spacetime. 

  Since quantum physics cannot be deduced from classical physics (at best, it should be the other way round), any so-called quantization of a classical theory 
is simply a concise summary of ad-hoc rules of procedure that lead from a classical system of equations to a system of quantum-physical equations 
after it has been discovered that these reproduce {\em some} empirical data for certain types of physical systems. 
  It is hard to imagine that any such ad-hoc procedure could have been proposed based on compelling reasoning, without the empirical success of the 
mix of partly ingenious, partly serendipitous, heuristic quantum physical proposals by Planck, Einstein, Bohr, de Broglie, Heisenberg, Born, Schr\"odinger, Dirac, etc.
 Even after-the-fact ``quantization procedures'' do not appear compelling in themselves;
Ed Nelson once summarized it thus: ``First quantization is a mystery, but second quantization is a functor.'' (Quoted in \cite{ReedSimon}, p.207).
 Yet once confirmed as a successful recipe, one may want to apply such a recipe to other classical systems and see what happens.
 That's precisely what many physicists have been doing for decades.
  Our paper may be perceived in this tradition.
  
  One reason for the ``mystery'' of orthodox (first) quantization procedures is that they change the mathematical meaning of the symbols used in the classical formulations, as a result of which {\em their deeper physical meaning has remained obscure} (paraphrasing Max Born).
  More to the point, it is not at all clear what those theories obtained by {\em canonical quantization}, and variations on its theme, say about nature in itself. 
 
  Yet not all quantization procedures are created equal. 
  
  Essentially in parallel with {\em canonical quantization}, an alternative approach had been conceived of, which does not suffer from such criticism,
for it in no way leaves the physical meaning of the resulting model obscure.
  For nonrelativistic quantum mechanics, this quantization procedure harks back to work of de Broglie in the mid-1920s; see in particular his presentation 
at the 1927 Solvay conference \cite{deBroglieSOLVAY}. 
  After being sidelined for 25 years, de Broglie's model was subsequently rediscovered and further developed by Bohm \cite{Bohm}, then Bell \cite{Bell}, and in the recent past it was advanced conceptually and mathematically rigorously by D\"urr, Goldstein, Zangh\`{\i} and their associates \cite{Equivar92}, 
\cite{DGMZ1999}, \cite{DT2009},  \cite{DGZbook}, \cite{CanBMrel14}.
  In a nutshell, the classical Newtonian dynamics of point particles is changed into a non-Newtonian dynamics of point particles.
  It is completely clear what this non-Newtonian theory of point particle motion, known as the de Broglie--Bohm theory, says about nature: {\em Matter consists of
  fundamental particles; their  locations in space are represented by mathematical points that move as time goes on}. 
  Since the transition from the Hamilton--Jacobi formulation of Newtonian classical point particle mechanics to the de Broglie--Bohm theory of (non-relativistic)
quantum dynamics of point particles can be effected by a technically minor (though conceptually radical) deformation of an action principle associated with the Hamilton--Jacobi theory, without losing the original meaning of the mathematical symbols, this transition has been called 
{\em least invasive quantization} \cite{K2}.
 In this sense it is also the ``least mysterious quantization.'' 
   
  We will briefly recall the key steps of the {\em least invasive quantization} procedure of Newtonian mechanics in the next section.
     After our condensed review we then turn to the main objective of this paper and inquire into a field analog of the least invasive quantization procedure of particle
 dynamics.
  More precisely, we investigate whether the DeDonder--Weyl--Christodoulou \cite{C1} Hamilton--Jacobi formulation of  {Lorentz-covariant} classical 
theory of fields $\phi: \cM \to \cN$ defined on a spacetime $\cM$ and taking values in a target space $\cN$ can be deformed into 
a quantum version analogously to what can be done for non-relativistic particle dynamics.
 A crucial issue is Christodoulou's integrability condition.
 Earlier investigations in this direction \cite{Freistadt}, \cite{Nik05}, \cite{Villani} were oblivious of it, but otherwise carried out in a similar spirit to ours.
 
 We identify a candidate quantum wave equation, namely Kanatchikov's Lorentz-covariant Dirac-type wave equation \cite{Kan99} on what Christodoulou termed 
configuration space $\cC := \cM \times \cN$ of the fields\stz{
\footnote{\stz{Note that this terminology is not the same as the usual meaning of ``field configuration space" in QFT.} \der{There, ``field configuration space" refers to the space of all possible configurations of a field on a spacelike hypersurface, and wave functions on ``field configuration space" are functionals of fields on spacelike hypersurfaces.}}}.
 We show that Kanatchikov's equation is well-posed and generally solvable.
 Then we establish the correspondence between plane-wave solutions of Kanatchikov's equation with solutions of the covariant Hamilton--Jacobi equations of  DeDonder--Weyl--Christodoulou in the classical limit. 
\stz{In particular, we prove the following result, stated informally here and 
\miki{explained} in subsequent sections of this article:} 
\vspace{-.5truecm}

\stz{
\begin{theorem}\label{thm:conv}
In the limit as the parameter $\lambda$ (analogous to $\hbar$) in Kanatchikov's equation \refeq{eq:Kanat} goes to zero, plane wave solutions of this equation, for a quantum wavefunction $\Psi$ guiding a Lorentz-covariant scalar field, satisfy the same equation as the plane wave solutions of the DeDonder-Weyl-Christodoulou  equations (\ref{eq:SFS},\ref{eq:SFw},\ref{eq:SFrho}), for the covariant Hamilton-Jacobi theory of the corresponding classical scalar field.
\end{theorem}
}

 This demonstrates that the DeDonder--Weyl--Christodoulou Hamilton--Jacobi equation captures the classical limit of Kanatchikov's wave equation for an important subset of solutions for which Christodoulou's integrability condition is satisfied.
  In follow-up work we hope to characterize their relationship in general.
  
 To be able to speak of a least invasively quantized Hamilton--Jacobi theory of fields, one also needs a corresponding guiding equation for these fields that evolves
 them from one spacelike slice to another, and that reduces to the Hamilton--Jacobi guiding equation in the classical limit.
 We propose a natural candidate for such a quantum guiding law and show that it yields \stz{
at each point $x$ on the spacelike slice, the existence of 
an $x$-dependent measure on the space of field values at that point
that is equivariant with respect to the flow induced by the guided field,}  \stznew{at least in the case of plane-wave solutions of Kanatchikov's equation:}

\begin{theorem}\label{thm:guiding}

\stznew{Let $\cM$ be the Minkowski spacetime and $\cN$ the space in which fields on $\cM$ take their values.}
\stznew{
\begin{enumerate}
\item
To every solution $\Psi$ of Kanatchikov's equation \refeq{eq:Kanat} there corresponds a foliation $\Sigma_t$ of $\cM$, a divergence-free vector field $\mathbf{Y} = (\mathbf{J},\mathbf{K})$ on the field configuration space $\cM\times\cN$, and 
for all $x\in \Si_t$ an $x$-dependent finite measure $\varrho_x$ 
on $\cN$ 
\newstz{with the following properties:}
\item  For fields $\phi:\cM\to\cN$ that are {\em guided} by $\mathbf{Y}$, i.e. that satisfy $\mathbf{J}\cdot \p \phi = \mathbf{K}$ on the graph of $\phi$, the flow on the configuration space induced by $\phi$ is embedded in the flow of $\mathbf{Y}$.
\item
When $\Psi$ is a plane-wave, the probability measure $\rho = \Psi^\dagger \Psi$ appearing in the usual statement of Born's Rule is simply the normalized version of $\varrho_x$, and it is equivariant in the usual sense of the word.
\end{enumerate}
}
\end{theorem}

We also show that this law is local in the sense of Einstein's special relativity. Our model therefore cannot be used to discuss Bell-type experiments that are known to violate Bell's inequality, something that one would expect a truly fundamental quantum theory of fields 
\stz{allow} one to do.
 
  The  DeDonder--Weyl--Christodoulou  Hamilton--Jacobi formulation of  {Lorentz-covariant} classical field
theory is recalled in section 3, and Kanatchikov's Lorentz-covariant Schr\"odinger equation and its relationship to the classical theory is investigated in section 4.   Our proposed guiding law for covariant fields is presented in section 5, and
in section 6 we conclude with a summary and an outlook.

  \section{Least Invasive Quantization of Newtonian Particle Mechanics}

 Our aim in this section is to explain least invasive quantization of Newtonian mechanics, in which a key role is played by the action principle.
 Thus we begin by recalling the main formulations of Newtonian mechanics that involve the action principle, and then we turn to the quantization step.
 
  The basic ideas can be explained by considering a single point particle with Newtonian inertial mass $m$ and given potential energy 
 $V = V(t,q)$, with $t\in\RR$ denoting an instant of time and $q \in \mathbb{R}^3$ a generic particle position in physical space.
    The generalization to the traditional many-body models with $N$ point particles interacting pairwise with various Newtonian forces is entirely straightforward.
   To go beyond the basic ideas we will also consider a charged, spinning point particle in $\RR^3$ that is acted on by external electromagnetic potentials $(\phi,{A})$.

     \subsection{Lagrange, Hamilton, and Hamilton--Jacobi}
  Our point of departure is the familiar Lagrange function
   $L(\,\cdot\,,\,\cdot\,,\,\cdot\,)  : \RR\times\RR^3\times\RR^3 \to \RR$ for the particle's motion, given by
\beq L(t,q,v) = \half m|v|^2 - V(t,q). \eeq
  Here,  $v\in T_q\RR^3\sim \RR^3$ denotes a generic velocity vector in the tangent space at $q$.
   An actual particle position at time $t$ will be denoted by $q(t)\in \RR^3$, the actual particle velocity at time $t$ is $\dot{q}(t)\in \RR^3$.
  For any differentiable path $t\mapsto q(t)$ that connects $q_1:=q(t_1)$ with  $q_2:=q(t_2)$ the corresponding action is
\beq \cA[q(\,\cdot\,)] := \int_{t_1}^{t_2} L(t,q(t),\dot{q}(t)) dt.\eeq
  We recall that the classical Newtonian motion is a stationary point of $\cA$ on the space of differentiable paths from $q_1$ to $q_2$.
  If the stationary point is  twice-continuously differentiable, it satisfies the Euler--Lagrange equation
 \beq \frac{d}{dt} \left(\frac{\p L (t,q(t),v)}{\p v}\Big|_{v=\dot{q}(t)}\right) - \frac{\p L (t,q,\dot{q}(t))}{\p q}\Big|_{q={q}(t)} =0 ,\eeq
or, explicitly, 
\beq  m \frac{d^2}{dt^2}q(t) = -  \frac{\p V(t,q)}{\p q}\Big|_{q={q}(t)}, \eeq
which is the familiar Newtonian equation of motion of a point particle in a conservative (though here time-dependent) force field, as found in introductory 
text books on classical mechanics.

 The path from the Lagrangian action principle to the Hamilton--Jacobi principle is via Hamilton's formulation.
 Defining the Hamiltonian to be the Legendre transform of $L$ as follows, 
 \beq\label{eq:HamLEGENDRE}
 H(t,q,p) := \max_v\{ \langle p,v\rangle - L(t,q,v)\} = \frac{1}{2m} |p|^2 + V(t,q), \eeq
where the angular parentheses denote duality pairing between tangent and co-tangent space at $q$.
 Here, (after, as usual for Euclidean space, identifying tangent and co-tangent spaces)  the generic canonical momentum $p$ emerges to be
\beq \label{def:p1}
p := \frac{\p \ell}{\p v} = m v.
\eeq
 Inverting the Legendre transform $L \mapsto H$ yields the map $H\mapsto L$, as
  \beq\label{eq:LagLEGENDRE}
 L(t,q,v) := \max_p\{ \langle p,v\rangle - H(t,q,p)\} = \frac{m}{2} |v|^2 - V(t,q);
 \eeq
now inserting this definition of $L$ into the Lagrangian action principle yields Hamilton's equations,
\beq\label{eq:HeqOFmot}
 \dot{q}(t) = \frac{\p H}{\p p}\Big|_{q={q}(t), p = p(t)} ,\qquad
\dot{p}(t) = - \frac{\p H}{\p q}\Big|_{q={q}(t), p=p(t)},
\eeq
which state that the velocity of the actual phase-space point $(q(t),p(t))$ at time $t$ is given by evaluating a generic phase-space velocity field 
$(\frac{\p H}{\p p}, - \frac{\p H}{\p q} )$ at that actual phase-space point.
 Understood in this way the system \refeq{eq:HeqOFmot} is a  {\em guiding law} on phase space, with a given guiding field.

 Hamilton's equations \refeq{eq:HeqOFmot}, written explicitly as
\beq\label{HeqOFmot}
 \dot{q}(t) = \frac{1 }{m}p(t)\qquad,\qquad
\dot{p} (t) =  f(t,q(t)), 
\eeq
here with $f(t,q) = - \p_q V(t,q)$, are a special conservative-forces case of Newton's equations of motion, who stated 
 \refeq{HeqOFmot} for more general forces $f(t,q,v)$. 
 This first-order system  \refeq{HeqOFmot} is manifestly equivalent to the second order formulation obtained from Lagrange's formalism.

 The Hamilton--Jacobi formulation emerges as a consequence of Hamilton's equations being form-invariant under canonical transformations to a new set 
 of variables $(Q,P)$, time-dependent functions of $(q,p)$.  
 This means that there are new variables
\beq\label{canon}
 Q =   Q(q,p,t)\qquad ,\qquad P =   P(q,p,t),
\eeq
 such that if $K$ denotes the Hamiltonian written in the new coordinates, then Hamilton's equations in the new coordinates read
\beq
 \dot{Q}(t) = \frac{\p K}{\p P}\Big|_{Q={Q}(t), P = P(t)} ,\qquad
\dot{P}(t) = - \frac{\p K}{\p Q}\Big|_{Q={Q}(t), P=P(t)}.
\eeq
Thus, $K$ is equivalent to $H$ --- in the sense that the corresponding Lagrangians $L$ and $\tilde{L}$ differ at most by a complete time derivative, i.e. along an actual motion we have
\beq
L - \tilde{L} =  \frac{d}{dt}\widehat{S}(t),
\eeq
where $\widehat{S}(t)= S(t,q(t),P(t),Q(t))$ for some function $S$ of $t$, $q$, $Q$, and $P$, evaluated at $q=q(t)$, $Q=Q(t)$, and $P=P(t)$. 
 This is because we can invert the canonical transformations \refeq{canon} to express $p = p(Q,P,t)$.

 The Hamilton--Jacobi formulation emerges when one realizes that the new coordinates can (in principle) be chosen such that $K \equiv 0$, 
 as a result of which $Q(t)$ and $P(t)$ will be constant: $\dot{Q} (t) = 0= \dot{P}(t)$.  
 We then have (suppressing the dependence of $S$ on $Q$ and $P$)
\beq
L(t) - \tilde{L}(t)  = \frac{\p S}{\p t} (t,q(t))+ \frac{\p S}{\p q}\Big|_{q=q(t)} \dot{q}(t),
\eeq
which can only be satisfied for all actual motions if, generically, 
\beq \label{eq:pH}
p = \frac{\p S}{\p q},\qquad -H = \frac{\p S}{\p t}.
\eeq
The second one of equations  \refeq{eq:pH} is the {\em Hamilton--Jacobi (partial differential) equation},
\beq 
{{\p_t S} + H(q, \nab_q S ,t) = 0}\,,
\eeq
a first-order nonlinear PDE for the {\em Hamilton--Jacobi phase function} $S(t,q) [= S(t,q; Q,P)$; here, $Q = (Q^1,Q^2,Q^3)$ and $P = (P^1,P^2,P^3)$ represent 
$6$ parameters].   $S$ is thus a function that depends on time $t$  and the generic position of the particle $q$. 
In the case of our particular example it's easy to see that the Hamilton--Jacobi equation has the form of an eikonal equation:
\beq\label{eq:HJeik}
\boxed{\p_t S  + \frac{1}{2m} |\nab_q S|^2 = -  V(t,q)}\,.
\eeq

 If more than one particle are present, say $N$, the phase function $S$ is a function of time $t\in \Rset$ and points $q\in \cC := \Rset^{3N}$,
 the {\em configuration space} of the $N$ point particles in $\RR^3$.

 Every solution  $S = S(t,q)$ of the Hamilton--Jacobi partial differential equation \refeq{eq:HJeik} gives rise to a 
$3$-parameter family of solutions of Hamilton's equations.  
 This is established through the method of characteristics, which requires solving the system of ODEs corresponding 
to the first of the Hamilton's equations 
 \refeq{HeqOFmot}, with the generic $p$ given by the first of equations \refeq{eq:pH}. 
 Thus, once a solution $S(t,q)$ for given initial data $S(0,q)$ has been obtained, the 
actual Newtonian particle motion with initial data $q(0)$ and $p(0) = \nab_q S(0,q)|_{q(0)}$ is obtained by solving 
\beq \label{eq:guiding}
\boxed{\dot{q}(t) = \frac{1}{m} \nab_q S(t, q)\Big|_{q = q(t)}}
\eeq
with $q(0)$ to be chosen. Note that once $q(0)$ is chosen, $p(0)$ is fixed by the first of equations \refeq{eq:pH}.
 The second one of Hamilton's equations \refeq{HeqOFmot} is automatically satisfied by solutions of \refeq{eq:guiding} by virtue of \refeq{eq:pH}. 
 
 Equation \refeq{eq:guiding} is a  {\em guiding law} on configuration space: the space gradient of the Hamilton--Jacobi phase function guides the actual motion of the particle by providing its velocity vector at the particle's actual location.
  For more than one particle this is a velocity vector field on configuration space, evaluated at the actual $N$-particle configuration.

 We already know that the ordinary differential equations of motion obey an action principle. 
  Also the Hamilton--Jacobi partial differential equation \refeq{eq:HJeik} is derivable from an action principle, in fact trivially 
  so at the expense of introducing a new function for the purpose of carrying out variations.
  In this vein, let $\varrho \geq 0$ be a function of $t$ and $q$, compactly supported on $\Rset^3$, then set
\beq
\cL_0(\varrho, S) := \left(\p_t S + V + \frac{1}{2m} |\nab_q S|^2 \right)\varrho,
\eeq
and let $\cA_0[\varrho,S] := \int_0^T\int_{\Rset^3} \cL_0(\varrho,S) d^3\!q dt$ be the corresponding action.  
 The stationary points of $\cA_0$ with respect to variations of $\varrho$ on $\varrho$'s support then trivially satisfy \refeq{eq:HJeik} on $\varrho$'s support.
 
 Now, if this was all there is to it, this variational principle for $S$ would merely be a ``cheap trick'' without offering anything new. 
 Yet if we ask for $\cA_0$ to be {\em jointly} stationary with respect to compactly supported variations of  {\em both} $\varrho$ and $S$,
 i.e. setting also the variation of $\cA_0$ with respect to ${S}$ equal to zero, we obtain the following continuity equation for $\varrho$:
\beq \label{eq:cont}
\boxed{\p_t \varrho + \nab_q\cdot \left( \varrho\frac{1}{m} \nab_q S\right) = 0}\,.
\eeq
Thus there is a natural interpretation that can be given to a stationary point $(\varrho,S)$ of $\cA_0$.

 Namely, the continuity equation implies that  $\int_{\Rset^3} \varrho(t,q) d^3\!q$ is preserved in time, hence
 normalizing $\varrho \geq	 0$ such that $\int_{\Rset^3} \varrho(t,q) d^3q =1$ we can at each time $t$ think of $\varrho$ as a probability density on configuration space
 $\cC$ for an ensemble of {\em independently evolving} mechanical systems of the same kind. 
  Given an initial map $q\mapsto p$ through $C^1$ initial data $S(0,q)$,
 by solving the equation \refeq{eq:HJeik} we obtain a time-evolving velocity field on configuration space, given by $\frac1m \nabla_q S(t,q)$, independently of
 the ensemble.
  Any initial $\varrho(0,q)$ defines an ensemble of  such {\em independent} individual systems, obeying the initial assignment $q\mapsto p$, and the 
  continuity equation for $\varrho(t,q)$ captures the change with time of the ensemble density as the velocity field on configuration space 
  transports the systems in the initial ensemble to other locations in configuration space, in the course of time. 

 Each individual system in the ensemble of course evolves according to \refeq{eq:guiding}, compatible with the evolution of the probability density. 
  This is our first instance of what is called {\em equivariance}.

\subsection{Least invasive quantization}
 In the Hamilton--Jacobi theory of evolution of a  classical ensemble of systems the ensemble density $\varrho$ is a {\em passive} variable, in the 
 sense that it is transported by the velocity field $\frac1m\nabla_q S$ while the evolution of $S$ in turn is completely unaffected by the evolution of
  $\varrho$ --- perfectly in line with the fact that the classical evolution of the systems in the ensemble is independent of each other. 
  Thus, once $S(t,q)$ has been determined, any initial $\varrho(0,q)$ is allowed and will be evolved by the same continuity equation.
   In particular, one can let $\varrho(0,q)\to \delta(q-q(0))$ and obtain the weak formulation of the guiding law \refeq{eq:guiding} for the actual motion of
   a single system starting at $q(0)$ with initial velocity $\dot{q}(0) = \frac1m \nab_q S(0,q)_{q=q(0)}$.
 
 The spirit of {\em least invasive quantization} is to change the role of $\varrho$ from being passive into {\em active}, a dynamical variable on par with $S$. 
  The continuity equation \refeq{eq:cont} and the guiding equation \refeq{eq:guiding} remain the same, but \refeq{eq:HJeik} will be continuously
  deformed (with the help of a parameter) into a modified PDE that also involves $\varrho$. 
   In the realm where Newton's mechanics gives an accurate description of nature, the $\varrho$-containing term will make an insignificant contribution and
   the modified Hamilton--Jacobi PDE will be practically (though not exactly) identical to \refeq{eq:HJeik}.
   
  Thus, while $\varrho $ still satisfies the mathematical requirements of a probability density, i.e. $\varrho\geq 0$ and $\int_{\Rset^3}\varrho(t,q)d^3\!q =1$,
  and therefore can be used for the computation of answers to all the same questions for which it can be used in the classical setting,
  it {\em gains} status beyond the original one by also assuming an active role in the generation of the particle dynamics.
   
  Of course, this idea in itself leaves open infinitely many possible ways to deform the Hamilton--Jacobi dynamics. 
   Yet there is only one way, modulo equivalences, to arrive at Schr\"odinger's equation on configuration space. 
    In the following we describe two equivalent procedures, each with its own narrative. 
     We also briefly explain the generalization that leads to the Pauli equation for a spinor wave function that guides the motion of a charged point particle.
         
   \subsubsection{Adding an entropy penalty term to the action functional  $\cA_0[\varrho,S]$.}
 The first procedure we describe starts from the Hamilton--Jacobi formulation involving the pair $(\varrho, S)$ and adds to the action functional the Fisher entropy functional of $\varrho$, multiplied by an appropriately dimensional parameter.
  Thus, the least invasive quantization  step  is the replacement
  \beq\label{eq:Qaction}
   \cA_0[\varrho,S] \to {\cA}_\hbar[\varrho, S] := 
   \int_0^T\!\! \int_{\Rset^3} \!\! \left(\p_t S + V + \frac{1}{2m} |\nab_q S|^2 \right)\!\varrho d^3\!q dt + \frac{\hbar^2}{2m}  \int_0^T\!\! \int_{\Rset^3}\! |\nab_q \sqrt{\varrho}|^2 d^3\!q dt .
  \eeq 
  Stationarity of $\cA_\hbar$ w.r.t.  variation of $S$ again yields the continuity equation \refeq{eq:cont}, while stationarity of $\cA_\hbar$ w.r.t. variation of  $\varrho$ now yields
\beq\label{eq:HJQF}
\boxed{\p_t S  + \frac{1}{2m} |\nab_q S|^2 = - V + \frac{\hbar^2}{2m} \frac{\De_q\sqrt{\varrho}}{\sqrt{\varrho}}} .
\eeq

 The least invasively quantized version of the classical Hamilton--Jacobi model, in which solutions of the Hamilton--Jacobi PDE \refeq{eq:HJeik}
 generate the actual motions through \refeq{eq:guiding} and transport an ensemble density of systems through \refeq{eq:cont}, describes a different 
 kind of dynamics. 
  Namely, the $\hbar$-deformed Hamilton--Jacobi PDE \refeq{eq:HJQF} has to be solved {\em jointly} with the continuity equation \refeq{eq:cont};
 subsequently the actual motions are obtained through solving the guiding equation \refeq{eq:guiding}. 
  
  We note that this model obeys equivariance of the evolution of the distribution of configurations. 

  The just described least invasively quantized Hamilton--Jacobi model is the de Broglie--Bohm model of quantum motions of point particles
  (modulo notational variations; 
  e.g. our $\sqrt{\varrho}$ is Bohm's $R$), but least invasive quantization is not how de Broglie or Bohm arrived at it. 
  Instead, their goal was to vindicate their (independent) proposals that solutions to Schr\"odinger's equation
\beq\label{eq:schrodinger}
i\hbar \p_t \psi = - \frac{\hbar^2}{2m}\De \psi + V \psi 
\eeq  
play the role of a {\em guiding field} on configuration space that generates the actual particle motions. 
 They noted that when they substituted the polar decomposition 
\beq \label{eq:polar}
 \psi := \sqrt{\varrho} e^{iS/\hbar}
\eeq
for $\psi$ in \refeq{eq:schrodinger} and sorted into real and imaginary parts,
one then easily sees that \refeq{eq:schrodinger} (locally) decomposes into the system \refeq{eq:HJQF}, \refeq{eq:cont}. 
  Since  \refeq{eq:HJQF} 
 superficially appears to be a Hamilton--Jacobi PDE with (what Bohm sanctioned) a ``quantum potential'' $-\frac{\hbar^2}{2m} \frac{\De_q\sqrt{\varrho}}{\sqrt{\varrho}}$ added to $V$, and since $\frac1m \nab_q S$ is still a velocity field on configuration space that transports the density $\varrho$ as per \refeq{eq:cont}, 
 it is compelling then to postulate that the velocity field also generates the actual motions through \refeq{eq:guiding}.
  De Broglie, in fact, was searching for a deformation of the Hamilton--Jacobi model already in his Ph.D. thesis, retaining \refeq{eq:guiding} all along,
  so for him pieces were falling in place when Schr\"odinger's equation appeared. 
  
\begin{remark}
There exists an unfortunate but widespread misperception that ``Bohm tried to derive quantum mechanics from Newtonian mechanics by adding new
kinds of forces that are very contrived.''
Such and similar statements, which can be found in the Encyclopedia Britannica and other authoritative publications, miss the point that the velocity field $\frac1m\nab_q S$
cannot be eliminated from the de Broglie--Bohm configuration space formulation of quantum mechanics to arrive at a genuine Newtonian theory of motion that 
involves only the actual positions and velocities (or momenta), and forces that depend only on these.
The fact that in the de Broglie--Bohm theory $\varrho$ plays an active role in the generation of the dynamics, which is impossible in any Newtonian dynamics,
 makes this plain: in de Broglie--Bohm theory you can't evolve $S$ without $\varrho$ and its second derivatives, and to evolve $\varrho$ you need $\nab_q S$. 
  (This and several other misperceptions of the de Broglie--Bohm theory are addressed in \cite{KieFoP}.)
 $\square$
 \end{remark}
 
   Least invasive quantization as described above also has the advantage that it is entirely straightforward to see that the mathematical  limit $\hbar\to 0$ of the 
  least invasively quantized model will manifestly give us back the Hamilton--Jacobi model of a $\varrho$-ensemble of classical systems, each one of which evolving as per the guiding equation  \refeq{eq:guiding} compatible with equivariance.
   Indeed, simply fix the initial data $S(0,q)$ and $\varrho(0,q)$ and note that the quantum potential term in \refeq{eq:HJQF} vanishes with $\hbar\to 0$. 
   The guiding equation and the continuity equation remain unaffected by this limit.
   
    Of course, in nature $\hbar$ is fixed, and thus the physically more informative criterion for the validity of the classical approximation is not the mathematical 
    limit $\hbar\to 0$ but a comparison of 
    the two terms $\propto\frac{1}{2m}$ in \refeq{eq:Qaction}. Since $S$ has physical dimension of action, set $S:=\hbar \Phi$. Then 
the de Broglie--Bohm particle dynamics described by \refeq{eq:guiding} and \refeq{eq:HJQF} approximates classical mechanical dynamics described
by  \refeq{eq:guiding} and \refeq{eq:HJeik} whenever  $|\nab_q\Phi|^2 \gg |\nab_q \ln {\varrho}|^2$. 
 Corrections to the classical motions can be computed systematically by looking for solutions of \refeq{eq:HJQF}
  in the form of a formal power series in $\hbar$, thus giving rise to WKB or so-called semi-classical approximations, all the time retaining the guiding equation
  \refeq{eq:guiding}.

  \begin{remark}
 In the orthodox QM literature  (e.g \cite{Messiah}, sect.6;  \cite{Pauli}, chpt.A, sect.12) one typically finds claims that the $\hbar\to 0$ limit of 
 \refeq{eq:HJQF}, producing \refeq{eq:HJeik}, would, in concert with \refeq{eq:cont}, prove that QM converges to classical Newtonian mechanics, 
 because the ``fluid elements on configuration space (with density $\varrho(t,q)$ and infinitesimal volume $\delta{V}$)'' would, in the Lagrange picture 
 of fluid mechanics, obey Newton's equations of motion (with ``infinitesimal mass $\varrho(t,q)\delta{V}$'').
 However, such claims are misguided when at the same time it is denied, as unfortunately is the case in orthodox accounts of
QM, that an individual $N$-body system would, at time $t$, have an actual configuration $q(t)$ unless one measures it. 
 If the complete quantum state at time $t$ of an individual system is $\psi_t\equiv (\varrho,S)_t$, then this remains the case when $\hbar \to 0$, because there 
 is nothing in the mathematics of Schr\"odinger's equation which would require that for an individual system $\varrho (0,q) \to \delta(q-q(0))$ as $\hbar \to 0$.
 So, to obtain Newtonian mechanics from QM, in addition to $\hbar\to 0$ one has to {\em postulate} that at time $t$ a physical $N$-body system does 
 have an actual configuration $q(t)$, or one inherits the dilemma of Schr\"odinger's cat. 
  And since one has to postulate this for the purpose of obtaining Newtonian mechanics in
  the limit $\hbar\to 0$, then it is prudent to postulate it directly for $\hbar >0$, evolving as per \refeq{eq:guiding}; 
  that's precisely what de Broglie--Bohm theory does. $\square$
  \end{remark} 

  \subsubsection{Switching from the real pair $(\varrho,S)$ to a complex $\psi$, then linearizing the PDE for $\psi$}
  There is a variation on the theme of what we call least invasive quantization, presented recently by Mia Hughes \cite{Mia},
   which has the enticing charm of showing that the de Broglie--Bohm version of 
  quantum mechanics could have been discovered by serendipity long before Planck's discovery of the law of the black body spectrum.
 
  Indeed, having the Hamilton--Jacobi model of an ensemble of Newtonian mechanical systems of the same kind that at time $t=0$ have momenta
  $p=\nabla_q S(0,q)$, with $q$ distributed by $\varrho(0,q)$, the ensemble distribution evolving as per \refeq{eq:cont} for an $S$ evolving as per \refeq{eq:HJeik},
 and with each member of the ensemble evolving as per \refeq{eq:guiding}, it is conceivable that someone notices that by defining 
 a complex-valued field defined on the particle configuration space, $\psi:\cC\to \Cset$, by $\psi: = \sqrt{\varrho}e^{iS/a}$, where $a$ is a parameter 
with physical dimension of action (for a narrative that pretends to predate Planck's law it would be strange to use $\hbar$ at this point, of course), the two PDEs 
\refeq{eq:HJeik} and \refeq{eq:cont} of the Hamilton--Jacobi model can be combined into a single complex-valued PDE, thus
\beq\label{eq:schrodclassical}
i a \p_t \psi = - \frac{a^2}{2m}\De \psi + V \psi + \frac{a^2}{2m} \frac{\De |\psi|}{|\psi|} \psi,
\eeq  
while the Hamilton--Jacobi guiding law expressed in terms of $\psi$ becomes 
\beq \label{eq:dBB}
\dot{q}(t) = a \frac{\Im(\psi^*\nab_q \psi)}{\psi^*\psi} \Big|_{q=q(t)}.
\eeq

 Before coming to the least invasive quantization step,  it is noteworthy to register that
while \refeq{eq:schrodclassical} is locally equivalent to the pair of PDE \refeq{eq:HJeik}, \refeq{eq:cont}, the continuity equation can be extracted
from  \refeq{eq:schrodclassical} without polar decomposition of $\psi$.
 Thus multiplying 
  \refeq{eq:schrodclassical} with $\psi^*$ and the complex conjugate equation of  \refeq{eq:schrodclassical} by $\psi$, then subtracting the latter from
  the former, the continuity equation for $\varrho = |\psi|^2$ materializes in the following form, identical to the one in QM:
\beq \label{eq:Schrodcont}
 \p_t |\psi|^2 + a \nab_q \cdot \Im( \psi^*\nab_q\psi ) = 0.
\eeq

 Having this equivalent complex reformulation of the Hamilton--Jacobi model of a $\varrho$-ensemble in front of us,  the {\em least invasive quantization} step, 
 modulo identification of $a=\hbar$, consists in  purging the nonlinear term in \refeq{eq:schrodclassical}, arriving at Schr\"odinger's equation \refeq{eq:schrodinger}, 
 except for the parameter $a$ in place of $\hbar$.
  Mia Hughes argued this step would be plausible because the double slit experiment suggests the linear superposition principle, and this is the obvious way to
  implement it.
  
 \begin{remark}
  Independently of the linear-superposition-principle argument, someone may have considered the approximate Hamilton--Jacobi PDE obtained by purging
  the nonlinear term in  \refeq{eq:schrodclassical}, simply because linear PDE are easier to discuss.
  Conceivably the so-approximated Kepler problem would have been studied, solved, and the Bohr energy spectrum discovered ahead of Bohr's model. 
   It is not too far fetched to imagine that someone in the 1890s could also have noticed that the family of differences $\triangle E$ between two energy levels resemble the Rydberg formula and (ahead of Planck, hence of Bohr) proposed that the hydrogen frequencies $\nu \propto \triangle E$, thus extracting a value for $a$, discovering $\hbar$. 

  $\square$
\end{remark} 
   
The above {\em least invasive quantization} step can also be implemented via the action principle. Expressed in terms of $\psi$ and $\psi^*$ as a real expression,
the classical  Lagrangian ${\cL}_C[\psi,\psi^*]$ reads
 \beq \label{eq:ClassicalLag}
{\cL}_C(\psi,\psi^*) = \frac{a i}{2}\left( \psi^* \p_t \psi -  \psi \p_t \psi^*\right)  - \frac{a^2}{2m} \nab_q \psi \cdot \nab_q\psi^*- V \psi^*\psi 
+ \frac{a^2}{2m} \big|\nab_q \sqrt{\psi^*\psi} \big|^2, 
\eeq 
and the classical action ${\cA}_C[\psi^*,\psi] := \int_0^T \int_{\Rset^3} {\cL}_C(\psi,\psi^*) d^3\!q dt$.
For the variations, $\psi^*$ and $\psi$ are to be treated as independent; $\psi^*$ variation yields \refeq{eq:schrodclassical}, $\psi$ variation the 
complex conjugate of \refeq{eq:schrodclassical}.
 The least invasive quantization step, modulo identification of $a$, is the replacement ${\cA}_C[\psi^*,\psi]\to {\cA}_Q[\psi^*,\psi]$, this time consisting of 
 {\em dropping} the Fisher entropy of $|\psi|$ from ${\cA}_C[\psi^*,\psi]$; viz. $ {\cA}_Q[\psi^*,\psi] := \int_0^T \int_{\Rset^3} {\cL}_Q(\psi,\psi^*) d^3\!q dt$,
 with
 \beq \label{eq:quantumLag}
{\cL}_Q(\psi,\psi^*) = \frac{a i}{2}\left( \psi^* \p_t \psi -  \psi \p_t \psi^*\right)  - \frac{a^2}{2m} \nab_q \psi \cdot \nab_q\psi^*- V \psi^*\psi .
\eeq 
 If $\psi$ is a stationary point of the action $ {\cA}_Q[\psi^*,\psi] $ then it will satisfy the Schr\"odinger equation
\beq \label{eq:Schrod}
ia \p_t \psi = \hat{H} \psi,\qquad \hat{H} := -\frac{a^2}{2m}\De_q + V,
\eeq
as obtained by purging the nonlinear term in \refeq{eq:schrodclassical}.
 Since only the Fisher term has been purged from the classical action $\cA_C[\psi^*,\psi]$, while the guiding equation \refeq{eq:dBB} has been retained, 
 we have arrived at the de Broglie--Bohm model of quantum motions --- modulo identification of $a$ with $\hbar$.
 
 We emphasize that in the least invasive quantization narrative, whether the $(\varrho,S)$ version or the $\psi$ version, there is no room for arguing why the guiding equation \refeq{eq:guiding}, equivalently \refeq{eq:dBB}, should suddenly be purged. 
  Doing so would give Schr\"odinger's equation alone, hence his cat.
 
\begin{remark}
 While we have emphasized that the $\psi$ version of least invasive quantization has the advantage of an ``appearance of plausibility'' 
 (what we mean by ``discovering quantum physics by serendipity"), we should also emphasize that
 there is a price to be paid, namely it is not at all manifest why \refeq{eq:ClassicalLag} should be the classical limit of  \refeq{eq:quantumLag}.
 Superficially, the two formulations, using the real pair $(\varrho,S)$ vs. using the complex $\psi$, even seem to contradict each other in their logic:
 in the $(\varrho,S)$ formulation, the Fisher term is missing from the classical action and is added in the least invasive quantization step; in the
 $\psi$ formulation, the Fisher term is purged in the least invasive quantization step. 
  Yet the two classical actions are equivalent, and so are the two quantum actions. 
   The illuminating resolution of this apparent paradox is left for the reader to ponder. $\square$
\end{remark}

\subsubsection{Generalization: charged particles and spinor wave functions} 
  The classical (test particle) motion in $\RR^3$ of a point electron with mass $m$ and electric charge $-e$ that is acted on
by applied electromagnetic fields obtained from given potentials $\phi$ and $A$ (assuming for simplicity no non-electromagnetic forces are acting) 
is obtained by simply identifying $V(t,q) =-e\phi(t,q)$ and replacing $\nab_q S(t,q) \to \nab_q S(t,q) +\frac{e}{c}A(t,q)$
in the classical Hamilton--Jacobi formalism that has been recalled above.
  Thus the actual Newtonian particle motion with given initial data $q(0)$ and $p(0) = \nab_q S(0,q)|_{q(0)}+ \frac{e}{c}A(t,q(0))$
is obtained by solving 
\beq \label{eq:guidingEM}
\boxed{\dot{q}(t) = \frac{1}{m} \left(\nab_q S(t, q)\Big|_{q = q(t)}+ \frac{e}{c}A(t,q(t))\right)}
\eeq
with $q(0)$ given, where $S$, with $S(0,q)$ chosen such that $p(0) = \nab_q S(0,q)|_{q(0)}+ \frac{e}{c}A(t,q(0))$, solves the Hamilton--Jacobi PDE 
obtained from compactly supported $\varrho$-variations of the action  
 \beq\label{eq:actionEM}
  \cA_0[\varrho,S] = 
   \int_0^T\!\! \int_{\Rset^3} \!\! \left(\p_t S -e\phi + \frac{1}{2m} \big|\nab_q S + \frac{e}{c}A\big|^2 \right)\!\varrho d^3\!q dt  ;
  \eeq 
variations w.r.t. $S$ yield the continuity equation 
for $\varrho$ with velocity field $\frac1m\left(\nab_q S(t,q) +\frac{e}{c}A(t,q)\right)$.

  These so-called minimal coupling steps are of course well-known; they are essentially dictated by demanding gauge-invariance of the dynamical theory.
  
  Also known, though less well-known, is the Hamilton--Jacobi theory for the classical motion of a point electron with charge $-e$, mass $m$, and
  magnetic moment $\mu(t)$ that obeys the Euler law of evolution  $\dot\mu(t) = \frac{e}{2mc}\mu(t)\times B\big(t,q(t)\big)$, which implies that $|\mu(t)|$ is constant. 
   For the magnitude of the electron's magnetic moment a natural choice is what elsewhere we have called the {\em classical magnetic moment} of the
electron, $e^3/4\pi mc^2$, i.e. the product of the elementary charge and the so-called {\em classical electron radius}, divided by $\frac{1}{4\pi}$.
  Since the equations now become a bit unwieldy, we only write down the action functional. 
   Let $\vartheta$ and $\varphi$ denote the usual angles on $\SS^2$. 
   Then $\mu = - (e^3/4\pi mc^2)(\sin\vartheta \cos \varphi, \sin\vartheta \sin\varphi, \cos \vartheta)$. 
   With $\vartheta$ and $\varphi$ functions of $t$ and $q$, the action reads
  \beq\label{eq:actionEMmu}
  \cA_0[\varrho,S,\vartheta,\varphi] = 
   \int_0^T\!\!\! \int_{\Rset^3} \!\! \left(\p_t S -e\phi +\tfrac{e^2}{4\pi c}\cos \vartheta \p_t\varphi
   + \tfrac{1}{2m} \big|\nab_q S + \tfrac{e}{c}A + \tfrac{e^2}{4\pi c}\cos \vartheta \nab_q\varphi \big|^2 -\mu\cdot B\right)\!\varrho d^3\!q dt  .
  \eeq  
  
   Next, to get from the classical theory of motion of such a point electron with magnetic moment, to the de Broglie--Bohm theory  involving Pauli's
   equation by least invasive quantization, takes two ingredients: replace the classical quantum of action $e^2/c$ by Planck's $h$ in all 
   expressions that represent the magnetic moment of the electron, call this functional $\cA_0^\hbar$, 
   then add the Fisher entropy term 
 \beq\label{eq:FisherMU}
  \cF[\varrho,\vartheta,\varphi] = 
  \frac{\hbar^2}{2} \int_0^T\! \!\int_{\Rset^3} \!\sum_{k=1}^3 \big|\nab_q \sqrt{\varrho_k}\big|^2  d^3\!q dt  
  \eeq  
to $\cA_0^\hbar[\varrho,S,\vartheta,\varphi]$ to obtain the quantum action functional  $\cA_\hbar[\varrho,S,\vartheta,\varphi]$; here,
$\varrho_1 = \varrho \sin^2\frac{\vartheta}{2}\cos^2 \frac{\varphi}{2}$, 
 $\varrho_2 = \varrho \sin^2\frac{\vartheta}{2}\sin^2 \frac{\varphi}{2}$,
 and
 $\varrho_3 = \varrho \cos^2 \frac{\vartheta}{2}$.
 Now defining
 \beq\label{eq:SPINOR}
 \psi = \sqrt{\varrho} e^{iS/\hbar} \begin{pmatrix} \cos\frac{\vartheta}{2} e^{i\varphi/2} \\ i \sin \frac{\vartheta}{2}e^{-i\varphi/2}\end{pmatrix},
 \eeq
 the least invasively quantized action reveals itself as the action of the Pauli equation, viz.
 \beq\label{eq:PauliACTION}
 \widehat\cA_\hbar [\psi,\psi^\dagger]= 
 \int_0^T\! \!\int_{\Rset^3}\!\!\left( \frac{\hbar}{2i}\left( \psi^\dagger\p_t\psi -(\p_t \psi^\dagger) \psi \right) 
 + \frac{1}{2m} \left| \vec\sigma \cdot \left(-i\hbar\nab_q  + \frac{e}{c}A\right)\psi\right|^2 - e\phi |\psi|^2 \right)d^3\!q dt  ,
 \eeq
 where $\vec\sigma$ is the vector of the three Pauli matrices. 

  The reverse steps, going from  $\widehat\cA_\hbar [\psi,\psi^\dagger]$ to $\cA_\hbar[\varrho,S,\vartheta,\varphi]$, have been described in \cite{Reginatto}.
 
\section{Hamilton--Jacobi Theories for Fields}
Before we can talk about field-theoretical generalizations of the least invasive quantizaion procedure outlined in the previous section, we need to review the generalization of Hamiltonian dynamics and Hamilton--Jacobi theory to the case of Lorentz-covariant fields. Most of the material in this section are taken from the excellent survey of Kastrup \cite{Kas83} and the pioneering work of Christodoulou \cite{C1} on the subject.

\subsection{Lagrangian field theory}  
Consider the action 
\beq\label{def:action}
\cA[\phi,\cD] := \int_\cD L\circ \si_\phi 
\eeq
for $\cD$ a bounded open domain in the spacetime $(\cM,g)$, which we take to be an $m+1$-dimensional Lorentzian manifold, and $\phi: \cM \to \cN$ a (classical) field on spacetime that is smooth (at least $C^1$) on $\cD$, and is taking its values in some smooth manifold $\cN$ of dimension $n$. 

Moreover, $L$ is a {\em Lagrangian density}, i.e. an $m$-form-valued section of the {\em velocity bundle}
\beq
\cV := \bigcup_{(x,q)\in \cM\times\cN} \cL( T_x\cM, T_q\cN)
\eeq
and $\si_\phi$ is the section of $\cV$ that corresponds to $\phi$ (see below.)
Here $\cL(V,W)$ denotes the set of linear maps from vector space $V$ to vectorspace $W$. Given local coordinates $(x^\mu)$ on $\cM$ and $(q^a)$ on  $\cN$, let $\{\frac{\p}{\p x^\mu}\}$ and $\{\frac{\p}{\p q^a}\}$ denote corresponding bases for the tangent spaces $T_x\cM$ and $T_q\cN$, and let $\{dx^\mu\}$, $\{dq^a\}$ be the dual bases to those, for $T^*_x\cM$ and $T^*_q\cN$, resp.  A linear transformation $v\in \cL(T_x\cM,T_q\cN)$ can thus be expanded in these bases as $v = v_\mu^a (dx^\mu\otimes \frac{\p }{\p q^a})$ so that $(x^\mu, q^a, v_\mu^a)$ form a local system of coordinates for $\cV$. We refer to $v_\mu^a$ as the {\em canonical velocities}.

$\cV$ is clearly a bundle over the field configuration space $\cC := \cM \times \cN$, and therefore also a bundle over $\cM$ as well as one over $\cN$.  For any bundle $\cB$ over a base manifold $\cM$, let $\pi_{\cB,\cM}$ denote the projection onto the base.  Let $\wedge_p\cM$ denote the tensor bundle of $p$-forms on the manifold $\cM$.  Consider the pullback bundle $\pi_{\cV,\cM}^*\wedge_m\!\cM$. This is a bundle over $\cV$.  A {\em  Lagrangian density} $L$ is a section of this bundle.  Given a map $\phi:\cM\to \cN$ let $d\phi:T\cM \to T\cN$ denote its tangent map (also known as its derivative, or differential). It follows that $\si_\phi : \cM \to \cV$ defined by $\si_\phi(x) = (x,\phi(x),d\phi(x))$ is a section of $\cV$ as a bundle over $\cM$.  Hence for any Lagrangian density $L$ defined as in the above, $L\circ \si_\phi$ is a section of $\wedge_m\cM$, i.e. an $m$-form, which therefore can be integrated on a domain $\cD$ in $\cM$.  Moreover, since $\cM$ is here assumed to be a Lorentzian manifold with a metric $g$, it has a distinguished volume form $\eps = \eps[g]$, the one that is covariantly constant, and thus there exists a function $\ell = \ell(x,q,v)$ such that $L\circ \si = (\ell \circ \si) \eps$. We call $\ell$ the Lagrangian density function corresponding to $L$.

For the remainder of this paper we are going to assume that $\cN$ is a linear space, i.e. isomorphic to $\Rset^n$ and that $\cM$ is the $3+1$-dimensional Minkowski space.  This is done mainly to simplify the presentation, and most of what follows is generalizable to curved spacetimes and nonlinear targets.  (For instance, when $\cN$ is a nonlinear space, in order to proceed further we would need to fix a symmetric connection on $\cN$, and then we always need to prove that any result we obtain is independent of the choice of that connection. See \cite{C1} for details.)

As is well known, the Euler-Lagrange equations for critical points of the action \refeq{def:action} take the form
\beq
\p_\mu \frac{\p \ell}{\p v_\mu^a} - \p_a \ell = 0.
\eeq
Consider the particular case of the following Lagrangian:
\beq
\ell_{SF}(x,q,v) = \half v_\mu^a v^\mu_a - V(q),
\eeq
so that, when composed with a section $\si_\phi$ we obtain
\beq
(\ell_{SF} \circ \si_\phi) (x) = \half \p_\mu \phi^a \p^\mu \phi_a - V(\phi(x)).
\eeq
This is the so-called Lagrangian of a {\em covariant (Lorentz-)scalar field} (``scalar" here may be a bit of a misnomer since $\phi$ may have multiple components, but it is meant to distinguish this with the case of $\phi$ being a vectorfield or other tensorfield defined on $\cM$, which requires a more careful treatment, see \cite{C1}, Chap. 4.)

The Euler-Lagrange equations for stationary points of the scalar field action $\cA_{SF}[\phi,\cD] := \int_\cD \el_{SF} \circ\phi \eps$ are therefore
\beq
\square_x \phi^a - \p^a V(\phi) = 0,\qquad \square_x := \p_\mu \p^\mu,
\eeq
which is a coupled system of linear or semi-linear wave equations, depending on whether the dependence of $V$ on $\phi$ is quadratic or not.  One could also obtain {\em quasilinear} systems of wave equations here, if one considers the more general family of Lagrangians that are of the form $\ell(x,q,v) = f(\xi) - V(q)$ with $\xi = v_\mu^a v^\mu_a$. For example the so-called membrane equations (a.k.a. ``scalar Born-Infeld") belong to this family, with $f(\xi) = 1-\sqrt{1-\xi}$.

\subsection{Covariant Hamiltonian field theory of DeDonder and Weyl}

Because Hamilton's equations \refeq{HeqOFmot} are ordinary differential equations, where the only independent variable is time, at first sight it seems problematic to ask for a covariant Hamiltonian theory, since the specification of a time function on spacetime breaks the Lorentz covariance of the theory.  
On the other hand, soon after Hamilton's original contribution, and beginning with Volterra in 1890, mathematicians had begun to obtain generalizations of Hamiltonian theory to more than one independent variable (see \cite{Kas83} for the full history of this subject.) One of these was DeDonder, who in 1913 obtained such a generalization, as well as an associated Hamilton--Jacobi theory for variational problems with several independent variables.  Another significant contribution was by Caratheodory in 1929.  In 1930 DeDonder published a monograph on the subject, and after further discussion and analysis of DeDonder's theory by Weyl in 1934/35, this theory began to be referred to as the {\em DeDonder--Weyl} (DW) theory, which is the one we will describe below.

Let $\cA$ be a classical action of the form \refeq{def:action}.  Recall that the quantities $v_\mu^a$ are called {\em canonical velocities}.  By analogy with the classical Hamiltonian mechanics, let 
\beq \label{def:p}
p_a^\mu := \frac{\p \ell}{\p v_\mu^a}
\eeq
be the {\em canonical momenta}.  Just as canonical velocities can be described abstractly as fibers of the velocity bundle $\cV$, it is possible to define a dual object, called the {\em momentum bundle} $\cP$, whose fibers are the canonical momenta $p$ (More precisely, it is $*p$, the Hodge dual of $p$ with respect to the volume form $\eps$ of $\cM$, that form the fibers of $\cP$.)  The definition of the momentum bundle is
\beq
\cP = \bigcup_{(x,q)\in \cM\times \cN} \wedge_{m-1,1}(T_x\cM,T_q\cN),
\eeq
 where $\wedge_{m-1,1}(V,W) := \cL(W,\wedge_{m-1}V) = \cL(\wedge^{m-1}V,W^*)$ for any two vector spaces $V$ and $W$. In coordinates, this means that the fibers over $\cC$ are
\beq *p = p^\mu_a \eps_{\mu,\mu_1,\dots,\mu_{m-1}}dx^{\mu_1}\wedge\dots\wedge dx^{\mu_{m-1}}\otimes dq^a.
\eeq

One can view the definition \refeq{def:p} of $p$ as a mapping from $\cV$ into $\cP$ that takes $v \mapsto p(v) :=\p \ell/\p v$. This is called the {\em Legendre transform}.  We will restrict ourselves to those variational problems for which this mapping is invertible, i.e. the Legendre transform is non-singular.  (Singular cases can be handled by a further extension of DW theory, see \cite{Kas83}.)

Let $v =v(p)$ denote the inverse of the Legendre transform. The {\em DeDonder--Weyl Hamiltonian} $H_{DW}$ is then defined to be
\beq \label{def:HDW}
H_{DW} = H_{DW}(x,q,p) := v(p)_\mu^a p^\mu_a - \ell(x,q,v(p)).
\eeq
The {\em DeDonder--Weyl--Hamilton} (DWH) equations are then as follows: 
\bna
\p_\mu q^a(x) & = & \frac{\p H_{DW}}{\p p_a^\mu}(x,q(x),p(x)) \label{eq:DWq}\\
\p_\mu p^\mu_a (x) & = & -\frac{ \p H_{DW}}{\p q^a}(x,q(x),p(x)), \label{eq:DWp}
\ena
where $p(x) := \frac{\p \ell}{\p v}(x,q(x),dq(x))$ (see \cite{Kas83} or \cite{C1} for a derivation.)  

While the definitions of momentum in the case of one independent variable \refeq{def:p1} and several independent variables \refeq{def:p} seem identical, as do the definitions of the corresponding Hamiltonians \refeq{eq:HamLEGENDRE} and \refeq{def:HDW}, there are some crucial differences between these two Hamiltonian theories, apart from the fact that DWH equations are not ODEs but PDEs:
\begin{enumerate}
\item Unlike \refeq{def:p1}, the definition of canonical momentum \refeq{def:p} is {\em not} unique.  It turns out that there is quite a bit of freedom in the definition of $p_a^\mu$. One can in fact define $p_a^\mu := \frac{\p \ell}{\p v_\mu^a} - h^{\mu\nu}_{ab}v_\nu^b$ where the only restriction on functions $h_{ab}^{\mu\nu}=h_{ab}^{\mu\nu}(x,q,v)$ is that they are anti-symmetric in both lower and upper indices.  This is connected with the theory of {\em null Lagrangians} (see \cite{C1}, p. 125).  On the one hand this freedom can be used to extend DWH theory to the cases where the Legendre transform is singular.  On the other hand, this lack of uniqueness means that there are many different Hamiltonian theories that correspond to the same Lagrangian field theory, and one perhaps needs some other criteria to distinguish between them.  Certainly the DeDonder--Weyl theory presented here is the {\em simplest} such theory, since its $h\equiv 0$, but it may not be the ``best" such theory by some other criteria. For example it can be argued that Caratheodory's theory is superior to DW in some respects (see \cite{Kas83} for a detailed discussion.)
\item Another main difference between single and multiple independent variable Hamiltonian theories is that in the latter case, the first Hamilton equation \refeq{eq:DWq} has an {\em integrability condition}: Suppose there exists a smooth field $\phi : \cM \to \cN$ such that the section $(x,q,v)$ with $q^a = \phi^a(x)$,  $v_\mu^a = \p_\mu \phi^a(x)$ satisfies \refeq{eq:DWq}.  Then the equality of mixed partial derivatives of $\phi$ implies that  the right hand side of this equation, when evaluated on $q = \phi(x)$, has to satisfy an integrability condition
\beq \label{eq:IC}
 \p_\nu \frac{\p H_{DW}}{\p p_a^\mu} - \p_\mu \frac{\p H_{DW}}{\p p_a^\nu} = 0.
\eeq
Another way to phrase this integrability condition is to think in terms of the inverse of the Legendre transform $v = v(p)$.  Once $p$ is determined in some way, this mapping implies that $v$ is determined. If this $v_\mu^a = \p_\mu \phi^a$ for some field $\phi$, then one must have
$\p_\nu( v(p)_\mu^a) = \p_\mu (v(p)_\nu^a)$.  Expanding this gives
\beq \label{eq:integcondv}
\p_\mu v_\nu^a + \p_b v_\nu^a v_\mu^b - \p_\nu v_\mu^a - \p_b v_\mu^a v_\nu^b = 0.
\eeq
The above can also be written in the language of differential forms: Let $w^a := v_\mu^a dx^\mu$.  Then \refeq{eq:integcondv} is equivalent to 
\beq \label{eq:integcond}
dw^a +w^b \wedge D_b w^a = 0,
\eeq
(which is the way Christodoulou expresses it in \cite{C1}.  This is the second equation in what he calls the {\em Euler System}). Here $d$ denoted exterior differentiation in $\cM$ and $D$ is a covariant derivative on $\cN$ (with respect to an arbitrary symmetric connection.)

Clearly, the integrability condition \refeq{eq:integcond} is a serious restriction.  It could make it difficult to solve these equations.  In particular, as we shall see in the next section, the existence of these integrability conditions makes the DeDonder--Weyl generalization of Hamilton--Jacobi equations much less useful than the classical HJ theory.  
\end{enumerate}

\subsection{DeDonder--Weyl--Christodoulou covariant Hamilton--Jacobi theory}

Arguably the most obvious generalization of the first equation in \refeq{eq:pH} to the case of several independent variables is to assume the existence of a {\em vectorfield} $\bS = S^\mu(x,q) \frac{\p }{\p x^\mu}$ defined on the field configuration space $\cC$ in such a way that 
\beq \label{eq:DWHJp}
p_a^\mu = \p_a S^\mu.
\eeq
From \refeq{eq:DWp} one immediately gets that 
\beq \label{eq:DWHJS}
\p_\mu S^\mu + H(x, q, \p_a S^\mu) = 0.
\eeq
(We have dropped the ``DW" subscript from $H$ since that's the only type of Hamiltonian we will henceforth be talking about.) This is called the DeDonder--Weyl Hamilton--Jacobi (DWHJ) equation.  One notes that this is a single equation for a four component object, so  there seems to be quite a bit of freedom in choosing $\bS$.  However, let $v = v(p)$ be the inverse Legendre transform, and let $\bS(x,q)$ be any solution of \refeq{eq:DWHJS}.  By \refeq{eq:DWHJp}, its gradient (with respect to field variables) determines $p$, and thus $v$ through the inverse Legendre transform.  As described in the previous section, this immediately gives rise to the integrability condition \refeq{eq:integcond}, which should be viewed as another equation for $\bS$.
Solutions of the DWHJ equation \refeq{eq:DWHJS} that do not satisfy this integrability condition, do not correspond to actual stationary points of the original Lagrangian, and are therefore of limited use as far as the analysis of solutions relevant to the theory is concerned.  

Let us examine this issue in the context of the classical scalar field theory.  Recall that $\ell_{SF} = \half v_\mu^a v^\mu_a - V(q)$, and therefore
\beq
p_a^\mu = \frac{\p \ell_{SF}}{\p v_\mu^a} = v^\mu_a,
\eeq
so that 
\beq
H_{SF} = \half p_a^\mu p^a_\mu +V(q).
\eeq
The DWHJ equation for the scalar field is therefore
\beq \label{eq:DWHJSF}
\p_\mu S^\mu + \half \p_a S^\mu \p^a S_\mu + V = 0.
\eeq
Moreover, we have 
\beq\label{dwguiding}
v_\mu^a = p^a_\mu = \p^a S_\mu,
\eeq
so that $\bS$ must also satisfy the integrability condition
\beq
\p_\nu \left( \p^a S_\mu(x,\phi(x)) \right) - \p_\mu\left(\p^a S_\nu(x,\phi(x)) \right) = 0,
\eeq
which, after expansion and using that $\p_\mu\phi^a = \p^a S_\mu(x,\phi)$ gives the following nonlinear second-order system for $\bS$
\beq
\p_a\p_\nu S_\mu - \p_a\p_\mu S_\nu + \p_b S_\nu \p^b \p_a S_\mu - \p_b S_\mu \p^b\p_a S_\nu = 0,\qquad a=1,\dots,n,\quad 0\leq\mu<\nu\leq 3.
\eeq
Setting $w^a := \p^aS_\mu dx^\mu$ shows that $w$ must satisfy Christodoulou's second ``Euler system" equation \refeq{eq:integcond}.
Therefore a complete set of equations for the DWHJ vectorfield $\bS$ for the classical scalar field theory is
\bna 
\p_\mu S^\mu + \half \p_a S^\mu \p^a S_\mu + V & = & 0,\label{eq:SFS}\\
d w^a + w^b \wedge \p_b w^a & = & 0,\qquad w^a:= \p^a S_\mu dx^\mu.\label{eq:SFw}
\ena
\subsection{Christodoulou's density function}

 Suppose now we want to attempt to carry out the same  least-invasive quantization procedure that we did in the case of a mechanical system of particles
 in $\Rset^3$.
  Thus we need to begin by considering an action principle for the DWHJ equation \refeq{eq:DWHJS}. 
   The obvious candidate for an action for \refeq{eq:DWHJS} is
\beq \label{eq:actionDWHJ}
\cA := \iint_\cC \left(\p_\mu S^\mu + H(x,q,\nab_q \bS)\right) \rho\  d^4xd^nq,
\eeq
with $\rho = \rho(x,q)$ a function defined on the field configuration space $\cC = \cM \times \cN$.  
 Demanding that $\cA$ be stationary w.r.t. variations of $\rho$ manifestly gives the DWHJ equation \refeq{eq:DWHJS}, while demanding stationarity of $\cA$ 
 w.r.t. variations of  $S^\mu$ will give us a system of equations for $\rho$:
\beq \label{eq:rho}
\p_\mu \rho + \p_a \left( \rho \frac{\p H}{\p p_\mu^a} \right) = 0,\qquad \mu = 0,\dots,3.
\eeq
Once again we note that the above system has an integrability condition. As shown in \cite[Prop.~2.11]{C1}, the integrability for the $\rho$ system above is the same as the original integrability condition we have seen before, namely \refeq{eq:integcond}. 

The above function $\rho$ to our knowledge makes its first appearance in mathematical literature in the work of Christodoulou \cite[p. 99]{C1}, who called it the {\em density function}\footnote{This terminology is justified since it is easy to see from \refeq{eq:actionDWHJ} that $\rho(x,\cdot)$ has the physical dimensions of a density (i.e. $1/\mbox{Volume}$) on $\cN$.}, and proposed the equation it should satisfy, which is the same as \refeq{eq:rho}, except that he wrote it in differential forms language, and for the Hodge dual of $\rho$ (with respect to the volume form of $\cN$), as
\beq
d\rho + \delta( w \cdot \rho) = 0
\eeq
($\delta$ = exterior derivative on $\cN$).  Christodoulou called this ``the equation of continuity," and used it to show that the integral of $\rho$ over any domain $\Om$ in $\cN$ is invariant under the Hamiltonian flow (assuming that the integrability conditions are satisfied.) 

We remark that, similar to the case of the classical Hamilton--Jacobi theory, here the first DWH equation \refeq{eq:DWq} also has the interpretation of a {\em classical guiding law}, this time for a classical field beable $\phi(x)$, namely, once $\bS$ is known, evaluating \refeq{eq:DWq} on $q=\phi(x)$ results in an equation specifying all derivatives of $\phi$ at a point $x$ once the value of $\phi$ is known at that point:
\beq
\p_\mu \phi^a(x)= \frac{\p H}{\p p^\mu_a}\left(x,\phi(x),\nab_q \bS(x,\phi(x))\right).
\eeq
Since the integrability condition for the above is assumed to be satisfied (recall this is \refeq{eq:SFw}) the above equation has a unique solution in the neighborhood of any point $(x_0,q_0) \in \cC$.  

For example, in the case where $H$ is the classical scalar field Hamiltonian \refeq{eq:DWHJSF}, the above is simply
\beq \label{eq:ClassGuid} \p_\mu \phi^a(x) = \p^a S_\mu(x,\phi(x)),
\eeq
while the equation satisfied by $\rho$, \refeq{eq:rho}, is simply
\beq \label{eq:SFrho}
\p_\mu \rho + \p_a (\rho \p^a S_\mu ) = 0.
\eeq

As explained by Christodoulou \cite[p. 93]{C1}, a solution $\bS$ of the DWHJ equations (\ref{eq:SFS}--\ref{eq:SFw}) generates a {\em flow} 
\stz{on field configuration space in the following way:  Let $F: \cM\times \cC \to \cN$ be the mapping}
$F(x,x_0,q_0) = \phi(x)$
where $\phi$ is the unique solution to \refeq{eq:ClassGuid} with $\phi(x_0) = q_0$\stz{, and let $F_{(x,x_0)}: \cN \to \cN$ denote the associate (2-parameter) flow on $\cN$: $F_{(x,x_0)}(q_0) = F(x,x_0,q_0)$.  For fixed $x_0\in \cM$ and $\Omega_{x_0}\subset \cN$ such that $q_0\in \Omega_{x_0}$, let $\Omega_x := F_{(x,x_0)}(\Omega_{x_0})$.}

\stz{On the other hand, one can use} 
 the density function $\rho$ to define an $x$-dependent {\em mass function} on $\cN$: For any point $x\in \cM$ and domain $\Omega \subset \cN$, one defines the (dimensionless) function
\beq \mu_x(\Om) := \int_\Om \rho(x,q) d^nq.
\eeq
 Christodoulou showed \cite[p. 104]{C1} that the equation of continuity \refeq{eq:SFrho} for $\rho$ implies that this mass function is invariant under the DWHJ flow $F$, 
\stz{in the sense that the function $m(x) := \mu_{x}(\Om_x)$ is independent of $x$.}

 \stz{Furthermore, integrating \refeq{eq:SFrho} on all of $\cN$ and using the divergence theorem yields that (under appropriate decay conditions on the field) $\mu_x(\cN)$ is a constant, and therefore can be normalized to be equal to one. Thus $x \mapsto \mu_x$ assigns to a point $x\in \cM$ a {\em probability measure} on $\cN$ that is invariant with respect to the DWHJ flow.} We wish to see to what extent this can be replicated for a ``quantized" version of the DWHJ theory.

The complete set of DeDonder--Weyl--Christodoulou (DWC) equations for a covariant Hamilton--Jacobi theory of the classical scalar field therefore consists of \refeq{eq:SFS} and \refeq{eq:SFrho}, {\em together} with their common integrability condition \refeq{eq:SFw}.  
If we truly wanted to replicate the least-invasive quantization procedure in this context, we would have to first find an action principle for this combined set of equations (including the integrability condition), and then seek to deform {\em that} action into one for a complex (or hypercomplex)-valued field $\Psi = \Psi(x,q)$ defined on the field configuration space that satisfies a linear, Lorentz-covariant equation.  At the moment we do not know if the set of equations (\ref{eq:SFS},\ref{eq:SFw},\ref{eq:SFrho}) is even derivable from a Lagrangian.  

In light of this issue, we may choose to proceed differently, first by trying to guess what a linear, Lorentz-covariant equation for $\Psi$ could look like, and then asking what would the analog to \refeq{eq:polar} be of the polar (or more appropriately, Pr\"ufer) decomposition of $\Psi$ in terms of $\rho$ and $\bS$, to see
if anything resembling the DWCHJ equations emerges in this way. This in a nutshell is the program started by I. V. Kanatchikov in a series of papers beginning with \cite{Kan99}. We can also ask if the resulting equations for $\rho$ allow it to be interpreted probabilistically, and if not, what further assumptions need to be imposed in order to do so. These are the subjects of our next section.

\section{Kanatchikov's Covariant Generalization of Schr\"odinger's Equation}

Kanatchikov argued that in order to render the Schr\"odinger equation \refeq{eq:Schrod} Lorentz-covariant, it is sufficient to
replace the time derivative $\p_t$ that appears on its left, by the covariant version of it, namely
$\slashed\p= \gamma^\mu\p_\mu$, with $\{\gamma^\mu\}_{\mu=0}^3\subset M_4(\Cset)$ being a set of Dirac matrices, i.e. generators of the complexified spacetime Clifford algebra $\Aset = \mathop{Cl}_{1,3}(\Rset)_\Cset$:
\beq
\label{eq:Kanat}
i\lambda\slashed\p \Psi = \hat{H} \Psi,\qquad \hat{H} = - \frac{\lambda^2}{2}\De_q + V,
\eeq
where $\lambda$ is a parameter whose relation to $\hbar$ is to be determined. (It is interesting to note that Kanatchikov's wave equation can also be recovered from a covariant canonical quantization procedure developed by von Hippel \& Wohlfarth applied to DeDonder-Weyl's classical Hamiltonian field theory \cite{vHippWohl06}.)

Since this change makes the operator $\slashed\p$ itself an element of this algebra, it stands to reason that what it acts on, namely $\Psi$, would also belong to this algebra, i.e. $\Psi$ would be a section of a vector bundle over field configuration space, with $\Aset$-valued fibers: $\Psi = \Psi(x,q) \in \Aset$.  

If on the other hand $\Psi$ belongs to a subspace $\Bset$ of $\Aset$ that is {\em not} an ideal of $\Aset$, then \refeq{eq:Kanat} also implies that the extra components that may appear on the left-hand side of this equation, those that are not in $\Bset$, would have to be zero, since $\hat{H}$ is a scalar operator, and therefore the right-hand side will always be in $\Bset$.

Kanatchikov then takes $\Bset$ to be the (5-dimensional) subspace of $\Aset$ spanned by the matrices $\{I, \ga^0,\dots,\ga^3\}$.  Here $I$ denotes the $4\times 4$ identity matrix.  Thus every $\Psi\in\Bset$ can be written as
\beq \label{eq:psipluspsimu}
\Psi = \psi I + \psi_\mu \gamma^\mu
\eeq
for $\psi ,\psi_\mu \in \Cset$.  Note that $\Bset$ is not an ideal, since the product of two gamma matrices is not in $\Bset$.  

\refeq{eq:Kanat} and \refeq{eq:psipsimu} now readily imply that $\psi=\psi(x,q)$ and $\psi_\mu = \psi_\mu(x,q)$ must satisfy the following system
\bna
i\lambda \p_\mu \psi & = & \hat{H} \psi_\mu\label{eq:psimu} \\
i\lambda \p_\mu \psi^\mu & = & \hat{H} \psi \label{eq:psi}\\
\p_\mu \psi_\nu - \p_\nu\psi_\mu & = & 0.\label{eq:icpsimu}
\ena
(The last equation is the consequence of $\Bset$ not being an ideal.)  Below we shall see that this system is in fact self-consistent and solvable. In particular, the third equation is the integrability condition of the first equation. (Note that we are only considering potentials $V =V(q)$, so that the Hamiltonian $\hat{H}$ does not depend explicitly on $x$.)

\subsection{Solvability of the Kanatchikov System}
We now turn to the system of equations (\ref{eq:psimu}--\ref{eq:icpsimu}) satisfied by $(\psi,\psi_\mu)$, and show that this is solvable so long as the quantum Hamiltonian $\hat{H}$ is self-adjoint as an operator on field configuration space and that it does not depend explicitly on the spacetime point $x$. These conditions are for example satisfied by $\hat{H} = -\frac{\lambda^2}{2}\De_q + V$ for a variety of $x$-independent potentials $V = V(q)$.
\stz{
\begin{proposition}
The Kanatchikov system (\ref{eq:psimu}--\ref{eq:icpsimu}) is well-posed and uniquely solvable so long as $\hat{H}$ is self-adjoint and does not depend explicitly on $x$.
\end{proposition}
}

\begin{proof}

We begin by noting that, differentiating \refeq{eq:psimu} and using \refeq{eq:psi} we obtain
\beq
i \lambda \p^\mu \p_\mu \psi = \p^\mu(\hat{H} \psi_\mu) = \hat{H} \p^\mu\psi_\mu  = \hat{H}\left(\frac{1}{i\lambda} \hat{H}\psi \right) = \frac{1}{i\lambda} \hat{H}^2 \psi,
\eeq
so that $\psi$ is decoupled from $\psi_\mu$ and satisfies its own equation
\beq\label{eq:psiwave}
-\lambda^2 \dal \psi = \hat{H}^2 \psi
\eeq 
where $\dal := \p^\mu \p_\mu$ is the D'Alembertian (wave) operator on Minkowski space.

Furthermore, \refeq{eq:psimu} shows that, once $\psi$ is known, $\psi_\mu$ would be determined.
To see this in more detail, let us define a 1-form $\tilde{\psi}(x,q) := \psi_\mu(x,q) dx^\mu$. We observe that \refeq{eq:icpsimu} implies that
\beq
(d\tilde{\psi})_{\mu\nu} = \p_\mu\psi_\nu - \p_\nu\psi_\mu = 0
\eeq
so that by Poincar\'e's lemma there exists a scalar $\chi = \chi(x,q)$ with the property that $\tilde{\psi} = d\chi$, 
where $d$ is the exterior derivative opertor in $x$.
Thus 
\beq \label{eq:psiisgrad}
\psi_\mu = \p_\mu \chi.
\eeq

Equation \refeq{eq:psiwave} can be solved in a variety of ways, including by separation-of-variables.  
\stz{S}uppose now that \refeq{eq:psiwave} is solved, so that $\psi$ is known.  Substituting \refeq{eq:psiisgrad} into \refeq{eq:psimu} we then obtain
\beq \label{eq:chi}
i\lambda \dal\chi = \hat{H}\psi
\eeq
which is an inhomogeneous wave equation with a given right-hand-side, and again solvable by various methods.  Once $\chi$ is known, $\psi_\mu$ is determined by \refeq{eq:psiisgrad}.

In conclusion,
the Kanatchikov system  (\ref{eq:psimu}--\ref{eq:icpsimu}) is uniquely solvable.  Moreover, the above analysis shows that
in Kanatchikov's Ansatz \refeq{eq:psipluspsimu}, 
the quantities $\psi_\mu$ are determined by $\psi$ and do not represent additional degrees of freedom for the system.

\end{proof}

\subsection{Connection of Kanatchikov's equation to DeDonder--Weyl Theory}

Kanatchikov \cite{Kan99} also proposed an analog of \refeq{eq:polar} for a $\Psi$ in $\Bset$, namely
\beq \label{eq:KanatDecomp}
\Psi = R \exp\big( {i} \gamma_\mu S^\mu/{\lambda}\big),\qquad R, S^\mu \in \Cset .
\eeq
This readily implies that
\beq  \label{eq:cossin}
\psi = R \cos\Big(\frac{1}{\lambda}|\bS|\Big),\qquad \psi_\mu = i R \frac{\sin(\frac{1}{\lambda}|\bS|) }{|\bS|} S_\mu,
\eeq
where 
$|\bS| := \left( S_\mu S^\mu \right)^{1/2}.$

From here on we are going to make the additional assumptions (not made by Kanatchikov)  that $R$ and $S_\mu$ are real, and that $\bS$ is a timelike vectorfield.  These assumptions are made in order to ensure that one does not get exponentially growing terms in \refeq{eq:cossin} as $\lambda \to 0$. 

Kanatchikov then continues his investigation by plugging \refeq{eq:cossin} into \refeq{eq:Kanat} and proceeds to derive equations for $R$ and $S_\mu$.  He arrives at the surprising conclusion that in the limit $\lambda \to 0$ one does {\em not} recover the DWHJ equation \refeq{eq:DWHJS} for $S_\mu$ (Kanatchikov does not appear to be aware of Christodoulou's continuity equation \refeq{eq:SFrho} and does not mention any candidate for the $R$ equation.)  More precisely, he shows that in order for $\bS$ to satisfy a deformed version of the DWHJ equation, it would have to satisfy a whole set of other constraints, which in essence force $\bS$ to be equal to $(S,0,0,0)$, thus throwing things back to the classical Hamilton--Jacobi equation of particle mechanics \refeq{eq:HJeik} and not the DeDonder--Weyl generalization of it to fields. 

With hindsight, this outcome should not have been surprising: as we have seen, the DWCHJ theory has an integrability condition \refeq{eq:SFw} that needs to be imposed {\em in addition} to the equations for $\rho$ and $\bS$, while the Kanatchikov system {\em contains} its own integrability condition.  Thus it stands to reason that the reduction of this system under the Ansatz \refeq{eq:KanatDecomp} would have to contain more than just the DWHJ equation for $\bS$ and the continuity equation for $\rho = R^2$, but also at least some remnant of the integrability conditions.

Here we take a different approach, recalling that our proposed least \stz{invasive} quantization procedure involves the deformation of the {\em action} of the classical theory, not directly the equations.  We therefore begin by reducing the action corresponding to Kanatchikov's equation \refeq{eq:Kanat}, under the decomposition \refeq{eq:KanatDecomp} he proposed, and compare that to the action of the DWCHJ equations in the same variables, to see how one can arrive at one from the other.

Let $\Aset$ denote the complexified spacetime algebra $\mathop{Cl}_{1,3}(\Rset)_\Cset$. The elements of $\Aset$ are called {\em Clifford numbers}. For $\Psi \in \Aset$ we define its {\em Dirac adjoint} $\overline{\Psi}$ (also known as {\em Clifford reversion}) as
\beq\label{eq:DiracAdjoint}
\overline{\Psi} := \gamma^0 \Psi^\dag \gamma^0.
\eeq
Since $\Aset$ is isomorphic to the matrix algebra $M_4(\Cset)$ it is convenient to think of $\Psi$ as a $4\times 4$ complex matrix.   The Dirac gamma matrices $\ga^0,\dots,\ga^3$ are $4\times 4$ matrices satisfying the Clifford relations $\{\ga^\mu,\ga^\nu\} = 2 \eta^{\mu\nu} I$ where $\eta = (\eta_{\mu\nu})$ is the Minkowski metric  $\eta = \diag(1,-1,-1,-1)$.  Note that $\overline{\ga^\mu} = \ga^\mu$. A convenient basis for $\Aset$ is
\beq
\cB :=\{ I;\quad \ga^0,\dots,\ga^3;\quad \al^1,\al^2,\al^3;\quad \Si^1,\Si^2,\Si^3;\quad \ga^5\ga^0,\dots, \ga^5\ga^3;\quad \ga^5 \}
\eeq
where $\al^k := \ga^0\ga^k$; $\Si^k := i\ep_{lj}{}^k\ga^l\ga^j$, and $\ga^5 = i \ga^0\ga^1\ga^2\ga^3$. The {\em scalar part} $(\Psi)_{sc}$ of a Clifford number is the coefficient of $I$ in the expansion of $\Psi$ in any basis for $\Aset$.  Using the matrix isomorphism, we have
\beq
(\Psi)_{sc} = \frac{1}{4}\tr \Psi
\eeq
where $\tr$ denotes ordinary matrix trace.  We also introduce the non-degenerate bilinear form
\beq
\langle \Phi,\Psi \rangle = \iint_\cC  \left(\overline{\Psi}\Phi\right)_{sc}\ d^4x d^nq
\eeq
for $\Bset$-valued fields $\Phi,\Psi$ defined on $\cC$.

An action for Kanatchikov's equation \refeq{eq:Kanat} is
\beq\label{eq:KanatAction}
\scr{S}[\Psi,\overline{\Psi}] = \iint_\cC \frac{1}{4}\tr \left\{\frac{\lambda}{2i}\left( \overline{\Psi}\slashed\p \Psi - \overline{\slashed\p \Psi} \Psi \right) + \overline{\Psi}\hat{H}\Psi \right\} d^4x d^nq = \iint \ell_Q[\Psi,\overline{\Psi}] d^4xd^nq.
\eeq
Suppose $\Psi \in \Bset$ has a decomposition of the form 
\beq \label{eq:KanatchDecomp}
\Psi = \psi I  +\psi_\mu \ga^\mu = R\exp\left({i}\ga_\mu S^\mu/{\lambda}\right).
\eeq
Clearly, $\lambda$ must have the same physical units as $S^\mu$, which can be seen to have units of action {\em density}. Thus we set $$\lambda = \hbar \kappa$$ where $\kappa$ is a constant with dimension (length)${}^{-n}$.

Let us define
\beq \label{eq:rhozetau}
\rho = |R|^2,\qquad \zeta := |\bS| 
,\qquad u^\mu := {S^\mu}/{|\bS|},
\eeq
so that we have
\beq \label{eq:psipsimu}
\psi = \sqrt{\rho} \cos(\zeta/\lambda), \qquad \psi_\mu = i \sqrt{\rho} \sin(\zeta/\lambda) u_\mu.
\eeq
Substituting these into the action \refeq{eq:KanatAction} we obtain that
\bna\label{eq:DKaction}
\ell_Q[\rho,\zeta, \bu] &=& \rho\left( u^\mu \p_\mu \zeta + \half |\nab_q \zeta|^2 + V(q) \right) \\
&& \mbox{} + \lambda \rho \sin\left(\frac{\zeta}{\lambda}\right) \cos
\left(\frac{\zeta}{\lambda}\right) \p_\mu u^\mu  + \frac{\lambda^2}{2} \rho  \p_a u^\mu \p^a u_\mu
\sin^2\left(\frac{\zeta}{\lambda}\right) + \frac{\lambda^2}{2} |\nab_q \sqrt{\rho}|^2.\nonumber
\ena
In order to see the relationship with DWCHJ, let us substitute $S^\mu = \zeta u^\mu$ in the integrand of \refeq{eq:actionDWHJ}, to obtain the classical Lagrangian density function
\beq
\ell_C[\rho,\zeta,\bu] = \rho\left(  u^\mu \p_\mu \zeta + \half |\nab_q \zeta|^2 + V(q)  + \zeta \p_\mu u^\mu + \half \zeta^2 \p_a u^\mu \p^a u_\mu\right).
\eeq
Using \refeq{eq:psipsimu} we can express the above as a function of $\Psi$ and $\overline{\Psi}$, with some left over terms:
\bna
\ell_C[\Psi,\overline{\Psi}] &= & \frac{1}{4}\tr \left\{\frac{\lambda}{2i}\left( \overline{\Psi}\slashed\p \Psi - \overline{\slashed\p \Psi} \Psi \right) + \overline{\Psi}\hat{H}\Psi \right\} \nonumber\\
&&\mbox{} +\frac{\lambda}{2} \rho \ \p_\mu u^\mu \left(2\frac{\zeta}{\lambda} - \sin 2\frac{\zeta}{\lambda} \right) \label{term1}\\
&& \mbox{} - \frac{\lambda^2}{2} \left\{ |\nab_q \sqrt{\rho}|^2 -\rho (\p_a u^\mu \p^a u_\mu) \left(\frac{\zeta^2}{\la^2} - \sin^2 \frac{\zeta}{\la}\right)\right\},\label{term2}
\ena 
and therefore one can obtain the quantum Lagrangian density $\ell_Q$ in \refeq{eq:KanatAction} by simply {\em dropping} the last two terms \refeq{term1} and \refeq{term2} in the above expression for the classical $\ell_C$, in analogy with the narrative provided in Section 2 for arriving at the quantum particle Lagrangian \refeq{eq:quantumLag}.  

We can also see that by making restrictive assumptions on $\bu$, for example that it is a {\em constant} unit vector, one can get rid of the extra terms in (\ref{term1},\ref{term2}) so that only the Fisher term remains.  This leads us to consider {\em plane-wave} solutions in the two theories, which by definition are solutions in which the direction of the vectorfield $S^\mu$ is constant, so that the wave fronts of the Hamilton-Jacobi flow on the configuration space are planes.  

\subsection{Plane-Wave Solutions of DWCHJ}
Let $n = (n^\mu)$ be a fixed unit time-like vector: $n_\mu n^\mu = 1$.  Consider the Anstaz
\beq\label{eq:density_g}
S^\mu(x,q) = f(n_\mu x^\mu, q) n^\mu,\qquad \rho(x,q) = g(n_\mu x^\mu ,q),
\eeq
where $f = f(\sigma,q)$ and $g=g(\sigma,q)$ are smooth functions defined on a domain $\Om\subset \Rset\times\Rset^n$. The parameter $\sigma$ labels the hyperplane leaves of the foliation of the spacetime determined by $n_\mu$. These functions are also subject to the restrictions:
\beq\label{eq:fgconds}
f(\sigma,q)\geq 0,\quad \p_\sigma f(\sigma,q)\leq 0,\qquad g(\sigma,q)\geq 0\qquad  \forall (\sigma,q) \in \Om.
\eeq
Suppose that $(S^\mu,\rho)$ solves the DWCHJ equations (\ref{eq:SFS},\ref{eq:SFrho}).  Since 
\beq
w^b = \p^b S_\mu dx^\mu = \p^b f(n\cdot x, q) n_\mu dx^\mu
\eeq
it follows that $dw^b = 0$ and $w^c\wedge \p_c w^b = 0$ so that the integrability condition \refeq{eq:SFw} is trivially satisfied. Moreover, since 
\beq
\p_\mu S^\mu(x,q) = \p_\sigma f(n\cdot x, q),\quad \p_b S_\mu(x,q) = \p_b f(n\cdot x, q) n_\mu,
\eeq
and
\beq
\p_\mu \rho(x,q) = \p_\sigma g(n\cdot x, q) n_\mu,\quad \p_b\rho(x,q) = \p_b g(n\cdot x, q)
\eeq
it follows that the functions $f$ and $g$ satisfy the classical Hamilton--Jacobi and continuity equations of Hamiltonian mechanics, with $H(q,p) = \half|p|^2+V(q)$:
\bna
\p_\sigma f + \half \p_b f \p^b f + V &=& 0,\label{eq:redf}\\
 \p_\sigma g + \p_b (g \p^b f ) & =&  0.\label{eq:redg}
\ena

\begin{example}
As a simple example (see \cite{HollandBOOK} \S 6.6.1 for more details), consider a single scalar field with potential $V(q) = \half \om^2 q^2$ for $q \in \Rset$.  For this potential, substituting the Ansatz
\beq
f(\sigma,q) = A(\sigma) + B(\sigma) q + \half C(\sigma) q^2
\eeq
in \refeq{eq:redf} yields that 
\beq\label{sol:redf}
f(\si,q) = -\frac{\om}{2} \left( \frac{B_0^2}{\om^2} + q^2\right) \tan(\om(\si - \si_0)) + B_0 q \sec(\om(\si-\si_0)),
\eeq
for constants $\si_0, B_0 \in \Rset$. 
We may now obtain the corresponding solution to \refeq{eq:redg}, since for given $f$ this is a transport equation.  Let $g_0: \Rset \to \Rset$ be {\em any} probability density function on $\Rset$, i.e. 
\beq
g_0(q)\geq 0,\qquad \int_\Rset g_0(q) dq = 1.
\eeq
For the corresponding solution of \refeq{eq:redg} that satisfies $g(\si_0,q) = g_0(q)$ we easily obtain
\beq
g(\si,q) = g_0\left( q \sec(\om(\si-\si_0)) - \frac{B_0}{\om} \tan(\om(\si-\si_0))\right) \sec(\om(\si-\si_0)).
\eeq
The  solution of DWHJ equations \refeq{eq:SFS} that corresponds to \refeq{sol:redf} is as follows:
\beq 
S^\mu(x,q) = \left\{-\frac{\om}{2} \left( \frac{B_0^2}{\om^2} + q^2\right) \tan(\om n\cdot (x - x_0)) + B_0 q \sec(\om n\cdot(x-x_0))\right\} n^\mu
\eeq
for any $x_0\in \Rset^{1,3}$, $B_0\in\Rset$.  Moreover, since we know that the integrability condition \refeq{eq:SFw} is satisfied, there exists field $\phi=\phi(x)$ such that
\beq
\p_\mu \phi(x) = \frac{\p S_\mu}{\p q}(x,\phi(x)) = \left\{-\om \tan(\om n\cdot(x-x_0)) \phi(x) + B_0 \sec(\om n\cdot(x-x_0))\right\} n_\mu
\eeq
which has the familiar solution
\beq 
\phi(x) = q_0 \cos(\om n\cdot (x - x_0)) + \frac{B_0}{\om} \sin(\om n\cdot(x - x_0)).
\eeq
We note that
\beq 
q_0 = \phi(x_0),\qquad B_0 n_\mu = \p_\mu \phi(x_0)
\eeq
are the initial values for the field and its 4-velocity (which has a constant direction).

\hfill $\square$
\end{example}

\subsection{Plane-wave solutions of Kanatchikov's equation} 

We next examine the existence of plane-wave solutions of the Kanatchikov equation, i.e. the system of equations satisfied by stationary points of the action corresponding to \refeq{eq:DKaction}.  We note that since $\bu$ is a unit vector, its variations are constrained.  To turn those into unconstrained variations we introduce a Lagrange multiplier $\Lambda$, so that the action becomes
\bna
\cS[\rho,\zeta, \bu,\La] &=& \iint d^4x d^nq \left\{  \rho\left( u^\mu \p_\mu \zeta + \half |\nab_q \zeta|^2 + V(q) \right) \right. \\ && \mbox{} \left.+ \lambda \rho \sin\left(\frac{\zeta}{\lambda}\right) \cos
\left(\frac{\zeta}{\lambda}\right) \p_\mu u^\mu  + \frac{\lambda^2}{2} \rho  \p_a u^\mu \p^a u_\mu
\sin^2\left(\frac{\zeta}{\lambda}\right) + \frac{\lambda^2}{2} |\nab_q \sqrt{\rho}|^2 + \La (u_\mu u^\mu -1 )\right\} \nonumber
\ena

Taking variations with respect to $\rho$, $\zeta$, $u^\mu$ and $\La$, we obtain respectively
\bna
&u^\mu\p_\mu \zeta + \half \p_a\zeta\p^a \zeta + V +\frac{\lambda}{2}\sin\frac{2\zeta}{\lambda}\p_\mu u^\mu + \frac{\lambda^2}{2}\sin^2\frac{\zeta}{\lambda}\p_au^\mu \p^au_\mu -\frac{\lambda^2}{2} \frac{\De_q \sqrt{\rho}}{ \sqrt{\rho}}= 0&\label{eq:zetaKan}\\
&\p_\mu (\rho u^\mu) + \p_a(\rho \p^a \zeta) -\rho \cos\frac{2\zeta}{\lambda} \p_\mu u^\mu - \frac{\lambda}{2}\rho\sin\frac{2\zeta}{\lambda} \p_a u^\mu \p^a u_\mu = 0&\label{eq:rhoKan}\\
&- \rho \p_\mu \zeta + \La u_\mu +\frac{\lambda}{2} \p_\mu(\rho \sin\frac{2\zeta}{\lambda})  + \lambda^2 \p_a (\rho\sin^2\frac{\zeta}{\lambda} \p^a u_\mu ) = 0&\label{eq:uKan}\\
& u_\mu u^\mu = 1.&
\ena
Using the last two equations, $\La$ can be computed:
\beq \label{eq:laval}
\La = \rho u^\mu \p_\mu \zeta -\frac{\lambda}{2} u^\mu \p_\mu \left(\rho \sin\frac{2\zeta}{\lambda}\right) + \lambda^2 \rho \sin^2\frac{\zeta}{\lambda} \p^a u_\mu\p_a u^\mu,
\eeq
so that the full set of equations that are equivalent to the Kanatchikov system is (\ref{eq:zetaKan}--\ref{eq:uKan}), with $\La$ substituted for from \refeq{eq:laval}.
We are now in position to consider plane-wave solutions to the above.  Again let $n = (n^\mu)$ be a fixed unit future-directed timelike vector, and consider solutions of the above equations that are of the form
\beq \label{planewaveAnsatz}
\zeta = f(n\cdot x,q),\quad \rho=g(n\cdot x, q),\quad u^\mu \equiv n^\mu,
\eeq
with $f$ and $g$ as before. From \refeq{eq:zetaKan} and \refeq{eq:rhoKan} we immediately obtain
\bna 
\p_\sigma f  +\half \p_a f\p^a f + V - \frac{\lambda^2}{2}\frac{\De_q \sqrt{g}}{\sqrt{g}} & = & 0,\label{eq:f}\\
\p_\sigma g + \p_a( g \p^a f ) & = & 0\label{eq:g}
\ena
which are {\em precisely} the Hamilton--Jacobi and continuity equations deformed by the addition of the ``Fisher information" term to the action or ``quantum potential" term to the HJ equation that we have encountered before.  We can also verify that \refeq{eq:uKan}, which contains an integrability condition, is automatically satisfied, so this is indeed an exact solution of the \stz{Kanatchikov} system. Furthermore, equation \refeq{eq:g} has precisely the form of a continuity equation so that $g$ can be interpreted as a probability density at a single `time' labeled by the parameter $\sigma$.

\stz{
Finally, it is evident that as $\la \to 0$, equation \refeq{eq:f} coincides with \refeq{eq:redf}, while \refeq{eq:g} is identical with \refeq{eq:redg}.  This establishes the claim in Theorem~\ref{thm:conv}.
}
\begin{example}
To continue with the harmonic oscillator example, we consider the special case where the initial wave function is a Gaussian wave packet, with mean $\mu$ and variance $\alpha$, i.e. 
\beq
g_0(q) = \frac{1}{\sqrt{2\pi\alpha}} e^{-(q - \mu)^2/2\alpha} .
\eeq
In this case the quantum potential term in \refeq{eq:f} is quadratic in $q$ and can therefore be used to cancel the classical harmonic oscillator potential $V$, so that the Ansatz
\beq f(\si,q) = A(\si) + B(\si) q
\eeq
suffices to solve \refeq{eq:f}.  The corresponding 
solution
(For simplicity we have taken $\si_0 = 0$ and $B_0 = 0$) 
to (\ref{eq:f}-\ref{eq:g}) is
\bna
f(\si,q) & = & -\frac{\lambda}{2} \om \si - \frac{\om}{2} \left( 2 q q_0 \sin \om \si - \half q_0^2 \sin 2\om\si \right) \\
g(\si,q) & = & \frac{1}{\sqrt {2\pi {\lambda}/{2\om}}} \exp\left\{-\frac{(q - q_0\cos \om \si)^2}{2{\lambda}/{2\om}}\right\}.
\ena
We note that as $\lambda \to 0$ the probability distribution concentrates on the classical trajectory, i.e. 
\beq
\lim_{\lambda \to 0} g(\si,q) = \de(q - q_0\cos \om\si).
\eeq
\hfill$\square$
\end{example}

The plane-wave Ansatz \refeq{planewaveAnsatz} has a counterpart for the $\Psi$ field:
\beq\label{planewave}
\Psi_n(x,q) = \sqrt{g(n\cdot x, q)} \exp\left\{\frac{i}{\lambda} f(n\cdot x,q) \gamma_\mu n^\mu\right\} .
\eeq
This will be a plane-wave solution of \refeq{eq:Kanat} for any $n \in \Rset^{1,3}$ such that $n_\mu n^\mu = 1$.  Since this equation is linear, any superposition of solutions is a solution. It thus follows that for any solution pair $(f,g)$ of (\ref{eq:f},\ref{eq:g}) satisfying \refeq{eq:fgconds} the following is a solution of \refeq{eq:Kanat}
\beq
\Psi(x,q) = \int_{n_\mu n^\mu = 1} a(n) \sqrt{g(n\cdot x,q)}\exp\left\{ \frac{i}{\lambda} f(n\cdot x,q) \ga_\mu n^\mu \right\} d\varsigma
\eeq
where $d\varsigma$ is the surface measure on the unit hyperboloid $n_\mu n^\mu = 1$ and $a(n)$ is an arbitrary function on it.

\section{A Guiding Law for Covariant Scalar Fields}

\subsection{Equivariance}
Let us recall that the guiding law \refeq{eq:dBB} for Schr\"odinger's equation \refeq{eq:Schrod} goes hand-in-hand with the continuity equation \refeq{eq:Schrodcont} to ensure that the {\em equivariance property} holds: 

A) Let $\Om \subset \Rset^3$ be a region in the particle configuration space. Integrating \refeq{eq:Schrodcont} on $\Om$ and using the divergence theorem one obtains
\beq \label{eq:prob}
\p_t\int_\Om \varrho(t,q) d^3\!q = - \int_{\p \Om} \varrho(t,q)\bv_\psi(t,q)\cdot \bn d\sigma
\eeq
where 
\beq
\varrho := |\psi|^2,\qquad {\varrho}  \bv_\psi := {\hbar}\Im( \psi^*\nab_q \psi).
\eeq
In particular if $\Om = \Rset^3$ and $\psi$ decays to zero sufficiently rapidly as $|q|\to \infty$ then the right-hand-side of \refeq{eq:prob} is zero, so that the total integral of $\varrho$ is independent of time. This total integral, when finite, can be normalized to be one, which gives $\varrho = |\psi|^2$ the interpretation of a probability density function defined on the particle configuration space; recall our discussion in section 2.

B) Now let $\Sigma_0 \subset \{0\}\times\Rset^3$ be a region in the initial time slice of the spacetime.  Consider the flow on the particle configuration space generated by the guiding equation \refeq{eq:dBB}: $\Phi_t(q) = q(t)$ where $q(t)$ is the solution to the ODE \refeq{eq:dBB} with initial data $q(0) = q$.  Let $\Sigma_t  := \Phi_t \Sigma_0$ and define $\cD := \bigcup_{0\leq t\leq T} \Sigma_t$.  Note that the continuity equation \refeq{eq:Schrodcont} implies that the vectorfield $\bX := \varrho\frac{\p}{\p t} + \varrho \bv_\psi^a \frac{\p}{\p q^a}$ is divergence-free on $\Rset\times \Rset^3$, so that we have
\beq\label{eq:diveq}
 0 =\iint_\cD \mbox{div}\,\bX = \int_{\p \cD} \bX \cdot \bn d\sigma = \int_{\Sigma_T} \varrho(T,q) d^3q - \int_{\Sigma_0} \varrho(0,q)d^3\!q + \int_\cK  \bX \cdot \bn_\cK d\sigma,
\eeq
where $\cK = \bigcup_{0\leq t\leq T}\Phi_t(\p \Sigma)$ is the lateral surface of the cylindrical domain $\cD$, obtained by flowing $\p \Sigma$ under the flow $\Phi_t$.  Since the velocity field of particles is $\bv_\psi$ it follows that the particle current is tangential to the lateral surface $\cK$ so that the last integral in \refeq{eq:diveq} is zero.  We therefore have
\beq
\int_{\Phi_T(\Sigma)} \varrho(T,q) d^3q = \int_{\Sigma} \varrho(0,q)d^3\!q.
\eeq

{\em Conclusion}: if the initial position of the particle is randomly distributed according to the density $\varrho(0,\cdot)$, then the position of the particle at any later time $t$ is distributed by the density $\varrho(t,\cdot)$.  This is called equivariance \cite{Equivar92}. 

More generally, any guiding equation (whether for particles or for fields) needs to be connected with a conserved current, so that there exists an equivariant density on the configuration space. An equivariant density on the configuration space is, along with the {\em quantum equilibrium hypothesis}, necessary for showing how de Broglie--Bohm guiding laws reproduce the standard QM predictions for subsystems \cite{Equivar92}.

\stz{We note here that, in contrast to the classical case \refeq{eq:SFrho}, in the quantum case the equation \refeq{eq:rhoKan} satisfied by $\rho = R^2 = \Psi^\dagger\Psi$ is {\em not} a continuity equation, due to the presence of extra $u$-dependent terms, so that $\rho$ in general does not appear to be the equivariant distribution we expect to have in connection with a Born Rule.}

\subsection{Riesz tensor and system of conservation laws for Kanatchikov equations}
Consider now the Kanatchikov equation as well as its adjoint
\beq
i\lambda \ga^\mu \p_\mu \Psi = \hat{H}\Psi,\qquad -i\lambda \p_\mu \overline{\Psi} \ga^\mu = \overline{\hat{H}\Psi}.
\eeq
Multiplying the first equation by $\overline{\Psi}$ on the left and $\ga^\nu$ on the right, multiplying the second equation by $\Psi\ga^\nu$ on the right, and subtracting the two equations, we obtain
\beq
i\lambda\p_\mu\left( \overline{\Psi}\ga^\mu \Psi \ga^\nu\right) = \overline{\Psi} \hat{H}\Psi\ga^\nu - \left(\overline{\hat{H}\Psi}\right)\Psi\ga^\nu.
\eeq
Therefore
\beq
i\lambda\p_\mu \left(\overline{\Psi}\ga^\mu\Psi\ga^\nu \right)  = -\frac{\lambda^2}{2} \left(
\overline\Psi\De_q\Psi - \De_q\overline{\Psi}\Psi\right)\ga^\nu = - \frac{\lambda^2}{2} \left[ \p_a(\overline{\Psi}\p^a\Psi) - \p_a(\p^a\overline{\Psi}\Psi)\right]\ga^\nu .
\eeq
Thus
we have a conserved (Clifford-algebra-valued) current
\beq\label{eq:conscurr}
\p_\mu \left(\overline{\Psi}\ga^\mu\Psi\ga^\nu\right)+ \p_a \left(\frac{\lambda}{2i} \left( \overline{\Psi}\p^a\Psi - \p^a\overline{\Psi}\Psi\right)\ga^\nu\right) = 0.
\eeq
Let us now define the {\em Riesz tensor} \cite{Rie1946}
\beq\label{eq:Riesz}
T^{\mu \nu} := \left(\overline{\Psi}\ga^\mu\Psi\ga^\nu\right)_\stz{sc} = \frac{1}{4}\tr \left(\overline{\Psi}\ga^\mu\Psi\ga^\nu\right),
\eeq
and a current defined on field configuration space $\cC$
\beq\label{def:K}
K^{a\nu} := \frac{\lambda}{2i}\left[\left( \overline{\Psi}\p^a\Psi - \p^a\overline{\Psi}\Psi\right)\ga^\nu\right]_\stz{sc}.
\eeq
From \refeq{eq:conscurr} we then have
\beq\label{eq:RieszCons}
\p_\mu T^{\mu\nu} + \p_a K^{a\nu} = 0, \qquad \nu = 0,\dots,3.
\eeq
Our proposed guiding equation will be a consequence of the above conservation laws.  Note however that \refeq{eq:RieszCons} in its current form is still not suitable for our purpose since what we actually need is a positive density that satisfies a continuity equation.  The key observation is that \refeq{eq:RieszCons} is not just one but {\em four} independent conservation laws, and we can use these four to construct a single conserved {\em timelike and future-directed} current, whose time component will give \stznew{rise to an} equivariant density, \stznew{at least when $\Psi$ is a plane wave.}

But before we do that, let us show that \refeq{eq:RieszCons} can already be used to obtain a guiding law for the covariant field.  First we need to recall an important property of the Riesz tensor:
\begin{proposition}\label{prop:DEC}
 $T^{\mu\nu}$ is a symmetric rank-two tensor and it satisfies the Dominant Energy Condition, i.e. 
\begin{enumerate}
\item $T^{\mu}_{\nu}Y^\nu$ is \stz{future-directed and} causal (i.e. either timelike or null) whenever $Y^\nu$ is \stz{future-directed and} causal ;
\item $T^{\mu\nu}Y_\mu Y_\nu \geq 0$ whenever the vectorfield $\bY$ is 
\stz{time-like}.
\end{enumerate}
\end{proposition}
\begin{proof}
To see this, we first compute the Riesz tensor \refeq{eq:Riesz} using Kanatchikov's decomposition \refeq{eq:KanatchDecomp}. From the definition of the Dirac adjoint \refeq{eq:DiracAdjoint}, we have 

\beq
\overline{\Psi} =\gamma^{0}\Psi^{\dagger}\gamma^{0}=  \gamma^{0}\left[\psi^{\star}I+\left(\psi_{\mu}\gamma^{\mu}\right)^{\dagger}\right]\gamma^{0}\\
  = \left[\psi^* I+\psi_{\mu}^{\ast}\gamma^{\mu}\right].
\eeq
So the Riesz tensor becomes 
\begin{eqnarray*}
T^{\mu\nu} & = & \frac{1}{4}\mathrm{tr}\left[\overline{\Psi}\gamma^{\mu}\Psi\gamma^{\nu}\right]\\
 & = &\left(|\psi|^{2}-\psi_{\alpha}^{\ast}\psi^{\alpha}\right)\eta^{\mu\nu}+\left({\psi^{\ast}}^{\mu}\psi^{\nu}+{\psi^\mu}{\psi^*}^\nu\right).
\end{eqnarray*}
which establishes that $T$ is a symmetric rank-2 tensor. 
Now, using the polar decompositions in \refeq{eq:psipsimu}, the above becomes
\begin{equation}\label{eq:KanatchDecompRiesz}
T^{\mu\nu}=\rho\left[\cos^{2}\left(\frac{\zeta}{\lambda}\right)-\sin^{2}\left(\frac{\zeta}{\lambda}\right)\right]\eta^{\mu\nu}+2\rho\sin^{2}\left(\frac{\zeta}{\lambda}\right)u^{\mu}u^{\nu}.
\end{equation}
Recall that $u^\mu$ is a unit timelike vectorfield, $u_\mu u^\mu = 1$.  It follows that $T^\mu_\nu u^\nu = \rho u^\mu$, i.e., $u^\mu$ is an eigenvector of $T^{\mu\nu}$ with eigenvalue $\rho$.  Let $h^{\mu\nu} := \eta^{\mu\nu} - u^\mu u^\nu$ denote the projection onto the orthogonal complement of $u^\mu$ in the tangent space at a point $x$ is spacetime.  Then
\beq\label{eq:Tmunufluid}
T^{\mu\nu} = \rho u^\mu u^\nu  + \rho \cos\left(2\frac{\zeta}{\lambda}\right) h^{\mu\nu} .
\eeq
It follows that the other three eigenvalues of $T^\mu_\nu$ are all equal to $\rho \cos\left(2\frac{\zeta}{\lambda}\right)$, with eigenvectors that are orthogonal to $u^\mu$.  
The conclusion then follows since cosine in absolute value is less than or equal to one.
\end{proof}
\begin{remark}
In fact we can see that, as a linear transformation, $T^\mu_\nu$ is (generically) invertible, with
\beq
{T^{-1}}^\mu_\nu = \frac{1}{\rho}\sec\left(2\frac{\zeta}{\lambda}\right)\left\{ \de^\mu_\nu - 2 \sin^{2}\der{\left(\frac{\zeta}{\lambda}\right)} u^\mu u_\nu \right\}.
\eeq
\end{remark}

\subsection{A proposed guiding law for field beables}
 Let $t$ be a time function on \stz{the Minkowski} spacetime $\cM$ and let $\Si_t$ denote the foliation of spacetime by constant $t$-slices.  Thus $\Si_t$ are spacelike for all $t$ and $\cM = \cup_{t\in\Rset} \Si_t$. Let $\shadibs = (x^1,x^2,x^3)$ denote an arbitrary system of coordinates on $\Si_0$.  Thus $(t,\shadibs)$ is a coordinate system on spacetime.  Let $S_T := \cup_{0\leq t\leq T}\Si_t$ denote a slab in spacetime. 
 Thus its boundary $\p S_T = \Si_T - \Si_0$ is the union of two Cauchy surfaces.  

Let $\cH$ denote the space of wave functions $\Psi = \Psi(t,\shadibs,q)$ such that $\Psi$ is continuous in $t$ and $\Psi(t,\cdot,\cdot)$ is square-integrable on $\Rset^3\times\Rset^n$, i.e.
\beq
\|\Psi(t,\cdot,\cdot)\|_{L^2}^2 := \int_{\Rset^3}\int_{\Rset^n} \left(\Psi^\dagger \Psi\right)_\stz{sc} d^nq d^3\shadibs < \infty .
\eeq
In other words $\cH = C^0\stz{\left(\RR,L^2(\Rset^3\times\Rset^n)\right)}$.

We are now ready to use the conservation law \refeq{eq:RieszCons} to construct 
\stznew{a general guiding equation for fields, and an}
 equivariant density \stznew{in the special case of plane-wave solutions of Kanatchikov's equation, in analogy with how we did so in the case of particles and Schr\"odinger's equation (though of course there we didn't need to assume plane-wave wave functions).}  Let $\bX$ be any fixed, timelike future-directed, covariantly constant vectorfield on Minkowski spacetime, i.e $\nab_\mu X_\nu = 0$.  In subsection \stz{\ref{subsec:X}} we will show how to find such an $\bX$ that is purely determined by the wave function $\Psi$.  Let us define the currents
\beq\label{eq:jxvec}
J_\bX^\mu := T^{\mu}_\nu X^\nu,\qquad K_\bX^a := K^a_\nu X^\nu.
\eeq
 Consider the vectorfield $\bY:\cC\to\cC$ on field configuration space as
\beq\label{eq:Yvectorfield}
\bY := J_\bX^\mu \frac{\p}{\p x^\mu} + K_\bX^a \frac{\p}{\p q^a}.
\eeq
\stz{Later on we will establish that} this vectorfield is divergence free.

Let $\Om_0\subset \cN$ be a bounded domain with smooth boundary in the target space $\cN$ (where the field takes its values.)  Consider the {\em flow} $\Phi_\tau:\cC \to \cC$ on the field configuration space  generated by the vector field $\bY$: 
\beq
\Phi_\tau(x_0,q_0) = (x(\tau),q(\tau))
\eeq
where $x(\tau),q(\tau)$ solve the following ODE initial value problem
\beq\label{ODEs}
\left\{\begin{array}{rcl} \dot{x}^\mu(\tau) & = & J_\bX^\mu(x(\tau),q(\tau))\\
\dot{q}^a(\tau) & = & K_\bX^a (x(\tau),q(\tau)) 
\end{array} \right.\qquad ; \qquad
\left\{\begin{array}{rcl}
x^\mu(0) & = & x_0^\mu \\ q^a(0)  & = & q_0^a
\end{array} \right. .
\eeq  
Furthermore, since $J_\bX^\mu$ is timelike and future directed, $J_\bX^0 \geq \|J_\bX^i\|$, the above flow is equivalent to the temporal flow $\Phi_t: \Si_0\times \cN \to \Si_t \times \cN$ given by the ODE
\beq
\left\{\begin{array}{rcl} \frac{ds^i}{dt} & = & \frac{J_\bX^i}{J_\bX^0}(\shadibs(t),q(t))\\[5pt]
\frac{dq^a}{dt} & = & \frac{K_\bX^a}{J_\bX^0} (\shadibs(t),q(t)) 
\end{array} \right.\qquad
\left\{\begin{array}{rcl}
s^i(0) & = & s_0^i \\ q^a(0)  & = & q_0^a
\end{array} \right. .
\eeq  
  Let us then define the domain $\cD$ in field configuration space
\beq
\cD := \bigcup_{0\leq t\leq T} \Phi_t(\Si_0\times\Om_0) .
\eeq
Its boundary therefore consists of three hypersurfaces
\beq
\p \cD =  (\Si_0\times\Om_0) \bigcup \Phi_T(\Si_0\times\Om_0) \bigcup\left(\cup_{0\leq t\leq T} \Phi_t(\Si_0\times\p \Om_0)\right) =: \cT \cup \cB \cup \cL
\eeq
which we suggestively call the ``top'', the ``bottom", and the ``lateral" part of the boundary.

Hence
\beq\label{eq:divthm}
0 = \int_\cD \mbox{div}\bY = \int_{\p \cD} \bY \cdot \bn_{\p \cD} d\si = 
\int_{\cT} \bY \cdot \bn_{\cT} - \int_{\cB} \bY \cdot \bn_\cB + \int_\cL \bY \cdot \bn_\cL d\si .
\eeq
Since by definition the flow $\Phi_t$ is at every point tangential to the lateral surface $\cL$, the last integral in the above is zero, and we have
\beq\label{eq:TopEqBot}
\int_{\cT} \bY \cdot \bn_{\cT} = \int_{\cB} \bY \cdot \bn_\cB .
\eeq

Let $\phi:\cM \to \cN$, $x\mapsto q=\phi(x)$ be any covariant field beable defined on $\cM$.  Obviously
\beq
\frac{dq}{d\tau} = \frac{\p \phi^a}{\p x^\mu} \frac{d x^\mu}{d \tau}.
\eeq
Thus if we want the flow generated by the field beable to be embedded within the flow of the vector field $\bY$, and hence consistent with \refeq{eq:TopEqBot}, by \refeq{ODEs} we need
\beq\label{eq:ge}
\boxed{J_\bX^\mu(x,\phi(x))\p_\mu \phi^a (x) - K_\bX^a(x,\phi(x)) = 0}.
\eeq

We may view the above as a partial differential equation for $\phi(x)$:
Given a solution $\Psi(x,q)$ of the Kanatchikov system, the above is a first-order quasilinear PDE for the unknowns $\phi^a(x)$.  To see this clearly, let us go into a coordinate frame in which $\bX = (1,0,\dots,0)$. Denoting the corresponding coordinates by $x = (t,\shadibs)$, the above then becomes
\beq\label{eq:QLPDE1}
A(t,\shadibs,\phi)\p_t \phi^a +B^j(t,\shadibs,\phi) \p_j \phi^a +C^a(t,\shadibs,\phi) = 0,
\eeq
where
\beq
A(x,q) = T^0_0(x,q),\qquad B^j(x,q) = T^j_0(x,q),\qquad C^a(x,q) = - K^a_0(x,q).
\eeq
It is well-known that a PDE of the type \refeq{eq:QLPDE1} is locally solvable using method of characteristics: one thinks of the solution $q^a = \phi^a(t,\shadibs)$ as a surface $F^a(t,\shadibs,q) = q^a - \phi^a(t,\shadibs) = 0$ in the configuration space $\cC = \cM\times\cN$.  Then the equation \refeq{eq:QLPDE1} expresses the fact that the vectorfield $(A,B^j,C^a)$ is tangential to this surface.  If we therefore parametrize this surface by parameters $(\tau,\xi^j)$ and denote differentiation with respect to $\tau$ (at constant $\xi$) by a dot, we obtain the system of ODEs along the characteristics:
\beq \label{eq:charODEs}
\dot{t} = A(t,\shadibs,q),\qquad
\dot{s}^j = B^j(t,\shadibs,q),\qquad
\dot{q}^a = C^a(t,\shadibs,q) .
\eeq
Given initial values for the field $\phi$ on the Cauchy surface $t=0$, i.e. $\phi(0,\shadibs) = \phi_0(\shadibs)$ one can then pose a family of initial value problems for the above ODE system by setting, for each $\xi$,
\beq 
t(0,\xi) = 0,\qquad s^j(0,\xi) = \xi^j, \qquad q^a(0,\xi) = \phi^a_0(\xi).
\eeq
Evidently, this has a local solution in a neighborhood of the Cauchy surface, thereby determining the evolution of the field $\phi(t,\shadibs)$.  
We thus propose \refeq{eq:ge} as the guiding law for the field beable. 

\stz{
It is worth emphasizing that, unlike the usual deBroglie-Bohm guiding law of particle mechanics, \refeq{eq:ge} is a PDE, not an ODE.  In particular, the said characteristics may intersect after a short time (think of Burgers' equation), causing a shock formation.  Thus going beyond the local solvability stated here is a nontrivial task and presumably depends heavily on the choice of the action for the classical field one starts with.  
}
\subsection{The guiding law obeys Einstein locality} 

We now address whether or not the 
evolution equation \refeq{eq:ge} paired with the \stz{Kanatchikov} equations  implies 
Einstein locality for the field beable $\phi$. 
First we recall that
\refeq{eq:ge}, which we remind ourselves is required for equivariance to hold, is a quasilinear PDE, since $J_{\mathbf{X}}^{\mu}\left(x,\phi(x)\right)$ depends on
both $x$ and $\phi$ but no spacetime derivatives of $\phi$. 
As a result, \refeq{eq:ge} can be equivalently written as a system of ODEs along characteristics, as in \refeq{eq:charODEs}.

Moreover, since $J_{\mathbf{X}}^{\mu}\left(x,\phi(x)\right)$ is timelike future-directed, these equations also imply that the domain of dependence of a solution
$\left(t,\shadibs,\phi(t,\shadibs)\right)$ is just the backwards light cone of the spacetime point $(t,\shadibs)$, and the domain of influence is the forward light cone.
 So \refeq{eq:ge} yields a locally causal evolution of $\phi$.
 
To the extent that our theory is Einstein local, it is expected to be physically and empirically inequivalent to textbook quantum field theory and extant versions of de Broglie-Bohm quantum field theory \cite{WardQFT10}, both of which violate Einstein locality \cite{WardQFT10}, \cite{Bell}. In future work, we will address the question of whether our theory can be generalized to violate Einstein locality while preserving Lorentz-covariance at the level of both the wave equation and the guiding equation, and at the same time preserving equivariance of an associated probability distribution.

\subsection{Construction of a distinguished vector field}\label{subsec:X}
We now show how to find a vectorfield $\bX_\Psi$ that is constructed out of the wave function $\Psi$, or more precisely, depends only on the initial values of the wave function $\Psi$ on the initial hypersurface $\Si_0$. 

The construction of this vector field follows the same procedure as in \cite{KTZ2018}.  
Let
\beq\label{eq:xtilda}
\tilde{X}^\mu(t) := \int_{\Rset^n}\int_{\Si_t} T^{\mu\nu}(t,\shadibs,q) n_\nu d^3\shadibs d^nq .
\eeq
Note that $T^{\mu\nu}$ is quadratic in $\Psi$ and we have assumed that $\Psi(t,\cdot,\cdot)\in L^2(\Rset^3\times\Rset^n)$ so that the above integral is finite.  
\stz{Also note that plane-wave solutions of the form \refeq{planewave} are not integrable in space, so that $\tilde{X}$ is not defined for them.}

Let us integrate \refeq{eq:RieszCons} on $S_t\times\Rset^n$, for fixed $t>0$, and use the divergence theorem:
\beq
0 = \int_{\Rset^n}\int_{S_t} \left(\p_\mu T^{\mu\nu} + \p_a K^{a\nu}\right) d^3x d^nq = \tilde{X}^\nu(t) - \tilde{X}^\nu(0) + \lim_{R\to\infty} \int\limits_{|\shadibs|=R}\int\limits_{|\bq|=R} \left(T^{\mu\nu}N_\mu + K^{a\nu}N'_a\right) d\si .
\eeq
In the surface integral term, $N_\mu dx^\mu + N'_a dq^a$ is the unit normal covector to the lateral boundary.  This term is zero provided $\Psi$ has sufficient decay at infinity.  In that case it would \stz{follow} that $\tilde{X}^\nu(t) = \tilde{X}^\nu(0)$, i.e. that $\tilde{\bX}$ is a constant vectorfield defined on $\cM$.  

It thus follows (for example using the usual density argument that smooth compactly supported functions are dense in $\cH$) that the vectorfield $\tilde{\bX}_\Psi$ thus constructed is a future-directed, causal, and constant vectorfield everywhere in $\cM$.  Moreover $\tilde{\bX}_\Psi$ clearly depends only on the initial values of the wave function, i.e. on $\left.\Psi\right|_{\Si_0}$.  For {\em generic} initial values, therefore, $\tilde{\bX}_\Psi$ will be {\em timelike}, so that it has a nonzero length, and thus can be normalized in the following way:
\beq\label{eq:xfromxtilda}
X_\Psi^\nu := \frac{\tilde{X}_\Psi^\nu}{\eta(\tilde{X}_\Psi,\tilde{X}_\Psi)}.
\eeq
Note that $\bX_\Psi$ thus defined is not a unit vector.  Let $\hat{\bX}_\Psi$ denote the unit vector in the direction of $\bX_\Psi$.
We now define the currents $\bJ_\Psi$ and $\bK_\Psi$ to be
\beq\label{eq:JfromXandT}
J_\Psi^\mu := T^{\mu}_\nu X_\Psi^\nu,\qquad K_\Psi^a := K^a_\nu X_\Psi^\nu.
\eeq
It follows that $J_\Psi^\mu\frac{\p}{\p x^\mu} + K_\Psi^a\frac{\p}{\p q^a}$ is a conserved current on the field configuration space $\cC = \cM\times\cN$:
\beq\label{eq:conscurrJK}
\p_\mu J_\Psi^\mu + \p_a K_\Psi^a  = 0.
\eeq

Under this assumption, \refeq{eq:TopEqBot} would then imply:
\beq\label{eq:TopBot}
\int_{\cT} \bY \cdot \bn_{\cT} = \int_{\cB} \bY \cdot \bn_\cB = \int_{\Om_0}\int_{\Si_0} T(\bX_\Psi,\bn_{\Si_0}) d^3\shadibs d^nq.
\eeq
This relation also holds if we replace the general unit vectorfield $\bn_{\Si_0}$ by the fixed unit vectorfield $\hat{\bX}_\Psi$ which defines a foliation of hyperplanes orthogonal to it. With this replacement, we may define \stz{the quantity
\beq\label{eq:varrhoXX}
\varrho_{\Si_\bX}^{\phantom{X}} (x,q) := T(\hat{\bX},\hat{\bX}) = \frac{1}{\|\bX\|}T(\bX_\Psi,\hat{\bX}_\Psi),
\eeq
}
where $\Si_\bX$ is a leaf of the foliation determined by $\bX_\Psi$ and contains the spacetime point $x$. By Prop.~\ref{prop:DEC} it follows that 
$\varrho_{\Si_\bX}^{\phantom{X}} $ is non-negative. Furthermore, it follows from \refeq{eq:xtilda} and \refeq{eq:xfromxtilda} that \stz{we have
\beq
\int_{\Rset^n}\int_{\Si_\bX} \varrho_{\Si_\bX}^{\phantom{X}} (x,q) d^3\shadibs d^nq = \left(\int_{\Rset^n}\int_{\Si_\bX} T_{\mu\nu} \hat{X}_\Psi^\nu \hat{X}_\Psi^\mu d^3\shadibs d^nq\right) = \frac{1}{\|\bX\|}\eta(\tilde{\bX},\bX) = \frac{1}{\|\bX\|} <\infty,
\eeq
so long as $\bX$ is strictly timelike. Thus for all $x\in\Si_\bX$, $\varrho_{\Si_\bX}^{\phantom{X}}(x,\cdot) $ is a finite measure on the space $\cN$ of generic field values at $x$.} 
Equation \refeq{eq:TopBot} now implies that 
\stz{
\beq
\int_{\cT} \bY \cdot \bn_{\cT} = \|\bX\|\int_{\Phi_T(\Si_0\times\Om_0)} \varrho_{\Si_\bX}^{\phantom{X}} (t,\shadibs,q) d^3\shadibs d^nq = \|\bX\|\int_{\Om_0}\int_{\Si_0} \varrho_{\Si_\bX}(0,\shadibs,q) d^3\shadibs
 d^nq.
\eeq
}
There remains the question of how $\varrho_{\Si_\bX}^{\phantom{X}} $ is related to the $\rho$ that appears in \refeq{eq:rho} and subsequently throughout section 4. As we saw, the $\rho$ that appears there admits an interpretation as a \stz{parameter-dependent} probability density function on \stz{$\cN$} 
when we restrict the Kanatchikov system to the special case of plane-waves \stz{(i.e. \refeq{planewave}, with the parameter being $\sigma = n\cdot x$)}. By contrast, $\varrho_{\Si_\bX}$ is \stz{an $x$-dependent finite measure on $\cN$ in the general case, so long as $\Psi$ is square integrable on $\Sigma_\bX\times\cN$ (which excludes plane waves.)} In fact there exists a more general relationship between $\varrho_{\Si_\bX}^{\phantom{X}}$ and $\rho$\stz{, namely}
\stz{
\begin{proposition}
With the above definitions, we have
$$ \varrho_{\Sigma_{\mathbf{X}}}^{\phantom{X}} 
\geq 
\rho,
$$
with the equality achieved (in a limiting sense) whenever the underlying wavefunction $\Psi$ is a plane-wave solution of Kanatchikov's equation.
\end{proposition}
}
\begin{proof}
\stz{To see this, first} notice that the Riesz tensor \refeq{eq:Tmunufluid} has the form of the energy-momentum
tensor of a relativistic perfect fluid, where the term proportional to $h^{\mu\nu}$
is the pressure, the term proportional to $u^{\mu}u^{\nu}$ is the energy density, and $u^{\mu}$ plays the role of the fluid 4-velocity.
Expectedly, if we pick the frame in which $u^{\mu}=\left(1,0,0,0\right)$, it follows that 
\begin{equation}
T^{00}=\rho.
\end{equation}
The general relationship between $\varrho_{\Si_\bX}^{\phantom{X}} $ and $\rho$ is found by using \refeq{eq:KanatchDecompRiesz} in \refeq{eq:varrhoXX} to obtain
\bna
\label{rhovarrho}
\varrho_{\Sigma_{\mathbf{X}}}^{\phantom{X}} 
=  T^{\mu\nu}\hat{X}_{\mu}\hat{X}_{\nu} \stz{ = \rho \left(
\left[1-\cos\frac{2\zeta}{\lambda}\right]\eta(\bu,\hat{X})^2 + \cos\frac{2\zeta}{\lambda}\right)
\geq 
\rho.}
\ena
The last inequality is due to the reverse Cauchy--Schwarz inequality that holds between timelike vectors in Minkowski space: $\eta(v,w)^2 \geq \eta(v,v)\eta(w,w)$.
\stz{Furthermore, we can see that the inequality in \refeq{rhovarrho} becomes an equality if $\bu \equiv \hat{\bX}$, which is a constant unit vector, which implies that $\Psi$ was a plane-wave.  In other words,} 
the inequality \stz{in \refeq{rhovarrho}} is \stz{in some sense} saturated \stz{for plane-waves.}

\stz{To make that more precise, let $\Psi$ be a plane-wave solution of Kanatchikov's equation, and let $\zeta,\rho,u$ be as in \refeq{planewaveAnsatz}, so that $u^\mu $ is a constant unit timelike vector. Let $\Sigma_\si := \{x\in\cM\ |\ \bu\cdot x = \sigma\}$ denote the foliation by hyperplanes in $\cM$ that is dual to $\bu$.  Thus for fixed $q\in \cN$, $\rho(\cdot,q)$ is a constant on each leaf $\Si_\si$ of the foliation.   Integrating \refeq{eq:g} on $\cN$ and using the divergence theorem, we obtain that
\beq \p_\si \int_\cN g(\si,q) dq = 0.  \eeq
As a consequence, for fixed $q$, $\rho(\cdot,q)$ is constant everywhere on $\cM$, and therefore it can be normalized so that
\beq \int_\cN \rho(x,q) dq = 1\quad\forall x \in \cM. \eeq
}

\stz{
Now, if we cut off the initial data corresponding to this $\Psi$ outside a large ball $B_R(0) \subset \Sigma_0$, so that it becomes square integrable in space, and solve Kanatchikov's equation with the cut-off data, we obtain a solution $\tilde{\Psi}$ that coincides with $\Psi$ inside $D_R\times\cN$, the domain of influence of $B_R(0)\times \cN$, but will be square integrable in space, and thus it would have a well-defined vectorfield $\bX = \bX_{\tilde{\Psi}}$ and a corresponding $\varrho_{\Sigma_{\bX}}$,  such that $\hat{\bX} \equiv \bu$   in $D_R\times \cN$.   We would then have that, inside $D_R\times\cN$,
\begin{equation}
\varrho_{\Sigma_{\mathbf{X}}}^{\phantom{X}}=T^{\mu\nu}\hat{X}_{\mu}\hat{X}_{\nu} =T^{\mu\nu}u_{\mu}u_{\nu} =\rho,
\end{equation}
so that equality is achieved in \refeq{rhovarrho}.    
The claims then follow when we take the cut-off away by letting $R \to \infty$.
}
\end{proof}
\stz{To complete the proof of Theorem~\ref{thm:guiding}, we define the $x$-dependent  measure $\varrho_x$ in the statement of the Theorem to be $\varrho_{\Sigma_\bX}$ where $\Sigma_\bX$ is the leaf of the foliation generated by the vectorfield $\bX = \bX_{\Psi}$ that includes the spacetime point $x$.}

We can go one step further and ask, when $\Psi$ is a plane-wave, what does the proposed guiding equation \refeq{eq:ge} reduce to?  To answer this question, let us fix a unit timelike direction $u^\mu$ and let $(f,g)$ be a solution of (\ref{eq:f}-\ref{eq:g}).  Let $\zeta,\rho,u$ be as in \refeq{planewaveAnsatz}, and let $\Psi$ be defined via \refeq{eq:psipsimu}. \stz{As we argued in the above, by introducing a cut-off we can arrange it that $\hat{\bX}_\Psi = \bu$ inside a large region in space.} Let us go into a frame where $\bu=(1,0,0,0)$.  We have already seen that $T^0_0 = \rho$.  It also follows that $T^j_0 = 0$ for $j=1,2,3$. We also need to calculate $K_0^a$. We have $\Psi = \sqrt{g}e^{i f \bu\cdot\gamma /\lambda}$, and plugging this into \refeq{def:K} we get
\beq
K_\nu^a(x,q) = \sqrt{g(\bu\cdot x, q)} \p^a f (\bu\cdot x, q) \left((\bu\cdot \gamma)\gamma_\nu\right)_\stz{sc} ,
\eeq
so that $K_0^a = \sqrt{g} \p^a f$.  The guiding law \refeq{eq:ge} now reads
$ \sqrt{g} \p_t \phi^a - \sqrt{g} \p^a f(t,\phi) = 0.$ In other words,
\beq\label{eq:PWguiding}
\p_t \phi^a = \p^a f(t,\phi),
\eeq
\stz{which with hindsight is quite natural and easy to guess from the form of the continuity equation \refeq{eq:g} for plane-wave solutions of Kanatchikov's equation.}

Once $\Psi$ is not a single plane wave, but a superposition of them, since the proposed guiding law is nonlinear (note that the tensors $T^{\mu\nu}$ and $K^a_\nu$ are quadratic in $\Psi$) the interaction of the different modes will contribute to the guiding of the field $\phi$, as was already the case for the deBroglie-Bohm guiding law \refeq{eq:dBB}.  In this \stz{limited} sense our guiding equation is an \stz{analogue}  of \refeq{eq:dBB} \stz{for} covariant fields.

\section{Summary and Outlook}

In this article we have shown that\stz{, in the classical limit,} plane wave solutions of Kanatchikov's relativistic generalization of Schr\"odinger's equation  
\stz{satisfy the same equations as} plane wave solutions
of the DeDonder--Weyl--Christodoulou Hamilton--Jacobi formulation of classical covariant field theory, and that they obey Christodoulou's integrability condition. 
 As far as we know, previous investigations did not pay attention to this constraint. 
  We also formulated a \der{local} guiding law for covariant fields that evolves the actual fields by a solution of Kanatchikov's equation, and we proved the equivariance of
  the associated \stz{evolving measure,} \stznew{when that solution is a plane wave.} 
  
 In future work we hope to extend the class of admissible solutions beyond plane wave category, and to address the issue of the apparent lack of Bell-type nonlocality for our guiding law. 

 Of particular interest is also to consider fields not defined on spacetime but {\em taking values in} spacetime.  Covariant guiding laws for such fields could be used to develop a de Broglie--Bohm type dynamical theory for {\em non-point-like} objects, e.g. ring-like particles \cite{KTZ2016}, or strings evolving in spacetime \cite{WardQFT10}, \cite{Weingard95}.  This aspect will be explored elsewhere.

\section*{Acknowledgments}

We are very grateful to Sheldon Goldstein and the mathematical physics group at Rutgers University for helpful discussions. \stz{We thank the anonymous referees for their detailed reading of the paper and their helpful remarks.} Maaneli Derakhshani gratefully acknowledges funding as a Post-Doctoral Associate in the Department of Mathematics in the School of Arts and Sciences at Rutgers University, from 2019--2020. 


\end{document}